\documentclass[12pt,letterpaper]{article}
\usepackage[margin=1.2in]{geometry}
\usepackage{amsmath, amssymb}
\usepackage{times}
\usepackage[mathscr]{euscript}
\usepackage{graphicx}
\usepackage{color}
\usepackage[normalem]{ulem}
\usepackage{bm}
\usepackage{epstopdf}
\numberwithin{equation}{section}
\usepackage{mathrsfs}
\usepackage[round]{natbib}
\usepackage{subcaption}
\graphicspath{ {images/} }
 \usepackage[table]{xcolor}
\usepackage{longtable}
\usepackage[margin=1.2in]{geometry}
\usepackage{amsmath, amssymb}
\usepackage{amsthm}
\usepackage{times}
\usepackage[mathscr]{euscript}
\usepackage{graphicx}
\usepackage{color}
\usepackage[normalem]{ulem}
\usepackage{bm}
\usepackage{epstopdf}
\numberwithin{equation}{section}
\usepackage{mathrsfs}
\usepackage[round]{natbib}
\usepackage{subcaption}
\graphicspath{ {images/} }
 \usepackage[table]{xcolor}
\usepackage{longtable}
\usepackage{array}
\usepackage{relsize}
\usepackage{pdflscape}
\usepackage[margin=1.2in]{geometry}
\usepackage{amsmath, amssymb}
\usepackage{times}
\usepackage[mathscr]{euscript}
\usepackage{graphicx}
\usepackage{color}
\usepackage[normalem]{ulem}
\usepackage{bm}
\usepackage{epstopdf}
\numberwithin{equation}{section}
\usepackage{mathrsfs}
\usepackage[round]{natbib}
\usepackage{subcaption}
\graphicspath{ {images/} }
 \usepackage[table]{xcolor}
\usepackage{longtable}
\usepackage{array}
\newcolumntype{P}[1]{>{\centering\arraybackslash}p{#1}}
\usepackage{changepage}
\usepackage[affil-it]{authblk}
\usepackage{multirow, booktabs}
\usepackage{multicol}
\usepackage[nodisplayskipstretch]{setspace}
\usepackage{rotating}

\usepackage{eurosym}
\usepackage{textcomp}

\usepackage{bbm}

\newtheorem{theorem}{Theorem}[section]
\newtheorem{lemma}[theorem]{Lemma}

\newcommand{\xTilda}{\tilde{\bm{x}}}
\newcommand{\zZ}{Z^\star}
\newcommand{\zz}{z^\star}

\begin{document}
\title{Modeling  Frequency and Severity of Claims with the Zero-Inflated Generalized  Cluster-Weighted Models}

\author{Nikola Po\v cu\v ca$^*$,  Petar Jevti\' c$^{**}$, Paul D. McNicholas$^*$, and Tatjana Miljkovic$^{\dagger}$}

\date{\small $^*$Department of Mathematics and Statistics, McMaster University, Hamilton, Ontario, Canada.\\
$^{**}$Department of Mathematics and Statistics, Arizona State University, Phoenix, AZ, U.S.\\
$^{\dagger}$Department of Statistics, Miami University, Oxford, OH, U.S.}
\maketitle
\doublespacing
\small

\begin{abstract}

In this paper, we propose two important extensions to cluster-weighted models (CWMs). First, we extend CWMs to have generalized cluster-weighted models (GCWMs) by allowing modeling of non-Gaussian distribution of the continuous covariates, as they frequently occur in insurance practice. Secondly, we introduce a zero-inflated extension of GCWM (ZI-GCWM) for modeling insurance claims data with excess zeros coming from heterogenous sources. Additionally, we give two expectation-optimization (EM) algorithms for parameter estimation given the proposed models. An appropriate simulation study shows that, for various settings and in contrast to the existing mixture-based approaches, both extended models perform well. Finally, a real data set based on French auto-mobile policies is used to illustrate the application of the proposed extensions.

\end{abstract}
\textsc{KEL Classification:}  C02, C40, C60.\\
\section{Introduction}\label{sec:introduction}
A significant number of clustering methods have been proposed for sub-grouping the data in the area of computer science, biology, social science, statistics, marketing, etc. \cite{Ingrassia+Punzo+Vittadini+Minotti:2015} proposed a cluster-weighted model (CWM) framework as a flexible family of mixture models for fitting the joint distribution of a random vector composed of a response variable and a set of mixed-type covariates with the assumption that continuous covariates come from Gaussian distribution. CWMs with Gaussian assumptions have been proposed by \cite{Gershenfeld:1997}, \cite{Gershenfeld:Schoner+Metois:1999}, and \cite{Gershenfeld:1999} in a context of media technology. Some extensions of this class of models have been considered by \cite{Punzo+Ingrassia:2015}, \cite{Ingrassia+Minotti+Punzo:2014}, \cite{Ingrassia+Minotti+Vittadini:2012}, \cite{subedi13,subedi15}, and \cite{punzo17}. These clustering methods are somewhat lacking for modelling insurance data, e.g., high excess zeros for claim count, heavy-tail loss distribution, deductible, or limits.

Sub-grouping of insurance policies based on risk classification is a standard practice in insurance. The heterogenous nature of insurance data allows for explorations of many different techniques for sub-grouping risk. As a result, there is a growing number of papers in the area of mixture modeling of univariate and multivariate insurance data to account for heterogeneity of risk. \cite{Lee+Lin:2010}, \cite{Verbelen+Gong+Antonio+Badescu+Lin:2015}, and \cite{Miljkovic+Grun:2016} proposed mixture models for univariate loss data, and mixture modeling of univariate insurance data has been extended to the multivariate context. A finite mixture of bivariate Poisson regression models with an application to insurance ratemaking was studied by \cite{Bermudez+Karlis:2012}. A Poisson mixture model for count data was considered by \cite{Brown+Buckley:2015} with application in managing a Group Life insurance portfolio. Recently, \cite{risks_miljkovic} reviewed two complementary mixture-based clustering approaches (CWMs and mixture-based clustering for an ordered stereotype model) for modeling unobserved heterogeneity in an automobile insurance portfolio, depending on the data structure under consideration. 

In this paper, we extend the CWM family proposed by \cite{Ingrassia+Punzo+Vittadini+Minotti:2015} to allow for modeling of non-Gaussian continuous covariates and a zero-inflated Poisson (ZIP) claims data with excess zeros, which are commonly seen in the insurance applications. We consider a generalized cluster-weighted model (GCWM) as well as a zero-inflated GCWM. Two partitioning methods are considered with two separate expectation-maximization (EM) algorithms \citep{Dempster+Laird+Rubin:1977}. The EM algorithm is based on the complete-data likelihood, which encompasses the observed data together with the missing data and/or latent variables. The EM algorithm can be highly effective for maximum likelihood estimation when data is incomplete or is assumed to be incomplete. The first EM algorithm is for parameter estimation for the GCWM models, while the second EM is for parameter estimation for the GCWM. We show that the Bernoulli and Poisson GCWM accurately estimate the initialization of the EM algorithm for the zero inflated GCWM model. These models utilize individual claims data and should be useful in the areas of ratemaking and risk management.

This paper is organized as follows. The GCWM and ZI-GCWM approaches are discussed in Section~\ref{sec:model}, and parameter estimation is discussed in Section~\ref{sec:estmeth}. Then, our methodology is applied to real data on French automobile claims and an extensive simulation study is conducted (Section~\ref{sec:numapp}). This paper concludes with a discussion and some suggestions for future work (Section~\ref{sec:sim}).

\section{Methodology}\label{sec:model}

\subsection{Background}

Let $(\bm{X^{'}}, Y)^{'}$  be the pair of a vector of covariates  $\bm{X}$ and a response variable $Y$. Assume this pair is defined on some sample space $\Omega$ that takes values in an appropriate Euclidian subspace. Now, assume that there exists $K$ non-overlapping partitions of $\Omega$, denoted as $\Omega_1, \ldots, \Omega_K$.  \cite{Gershenfeld:1997} characterized CWMs as a finite mixture of GLMs; hence, the joint distribution $(\bm{X^{'}}, Y )^{'}$ has the form
 \begin{align}
 f(\bm x, y; \bm{\Phi})= \sum_{k=1}^{K} \tau_k q(y|\bm{x};\bm{\vartheta}_k)p(\bm{x};\bm{\theta}_k),
\label{eq1}
\end{align}
where $\bm{\Phi}:=\{\bm{\vartheta}_k, \bm{\theta}_k\}$ denotes the model parameters.
The pair $q(y|\bm{x};\bm{\vartheta}_k)$ and $p(\bm{x};\bm{\theta}_k)$ are conditional and marginal distributions of $(\bm{X^{'}}, Y)^{'}$, respectively, while $\tau_k$ is the $k$th mixing proportion such that $\sum_{k=1}^{K}\tau_k=1$, $\tau_k>0$.
\cite{Ingrassia+Punzo+Vittadini+Minotti:2015} proposed a flexible family of mixture models for fitting the joint distribution of a random vector $(\bm{X^{'}}, Y)^{'}$ by splitting the covariates into continuous and discrete, i.e., $ \bm{X}=(\bm{V}',  \bm{W}')'$. The assumption of independence between continuous and discrete covariates allows us to multiply their corresponding marginal distributions. Thus, for this setting the model in \eqref{eq1} is reformulated as follows
\begin{align}
 f(\bm{x}, y; \bm{\Phi})= \sum_{k=1}^{K} \tau_k q(y|\bm{x};\bm{\vartheta}_k)p(\bm{x};\bm{\theta}_k)=\sum_{k=1}^{K} \tau_k q(y|\bm{x};\bm{\vartheta}_k)p(\bm{v}; \bm{\theta}_k^{\star})p(\bm{w};\bm{\theta}_k^{\star\star})
\label{eq2}
\end{align}
where $\bm{v}$ and $\bm{w}$ are respectively the realized vectors of continuous and discrete covariates, $q(y|\bm{x};\bm{\vartheta}_k)$ is the conditional density of $Y|\bm{x}$ with parameter vector $\bm{\vartheta}_k$, $p(\bm{v};\bm{\theta}_k^{\star})$ is the marginal distribution of $\bm{v}$ with parameter vector $\bm{\theta}_k^{\star}$, and $p(\bm{w};\bm{\theta}_k^{\star\star})$ is the marginal distribution of $\bm{w}$ with parameter vector $\bm{\theta}_k^{\star\star}$. As before, $\bm{\Phi}$ denotes the set containing all model parameters. 
The conditional distribution $q(y|\bm{x};\bm{\vartheta}_k)$ is assumed to belong to an exponential family of distributions and as such can be modeled in the GLM framework. The marginal distribution of continuous covariates is assumed to be Gaussian. 

Unfortunately, this last assumption is too strong for use in insurance related applications, specifically in rate-making or reserving. To relax it, we develop an extension which allows for non-Gaussian covariates as discussed in the next section.

\subsection{Generalized Cluster-Weighted Model}
We proceed to extend \eqref{eq2} by splitting the continuous covariates $\bm{V}$ via $\bm{V}=(\bm U^{'}, \bm T^{'})^{'}$, where $\bm{U}$ contains the non-Gaussian covariates and $\bm{T}$ contains the Gaussian covariates. We do this to retain the possibility of having both Gaussian and non-Gaussian covariates. Then, with the final relabelling of parameter vectors, \eqref{eq2} becomes
\begin{align}
 f(\bm x, y; \bm{\Phi})= \sum_{k=1}^{K} \tau_k q(y|\bm{x};\bm{\vartheta}_k)p(\bm{t};\bm{\theta}_k^{\star})p(\bm{w};\bm{\theta}_k^{\star\star})p(\bm{u};\bm{\theta}_k^{\star\star\star}),
\label{eq3}
\end{align}
which we refer to as the generalized cluster-weighted model (GCWM). Here, $p(\bm{t};\bm{\theta}_k^{\star})$ denotes the marginal density of Gaussian covariates, with parameter vector $\bm{\theta}^{\star}_k$, and $p(\bm{u};\bm{\theta}_k^{\star\star\star})$ denotes the marginal density of the non-Gaussian covariates with parameter vector $ \bm{ \theta}_k^{\star\star\star} $.

Due to its relevance for the applications in the actuarial domain, in this work, we make the particular choice of the multivariate log-normal distribution  for non-Gaussian covariates --- this, however, does not reduce the generality of our developed framework. With the log-normal assumption for $p(\bm{u};\bm{\theta}_k^{\star\star\star})$, we have that $\bm{u}$ is defined on $\mathbb{R}^p_+$ with parameter vector $\bm{\theta}_k^{\star\star\star}= (\bm{\mu}_k^{\star\star\star} ,\bm{\Sigma}_k^{\star\star\star})$ and probability density function
\begin{equation}\label{eqn:logn} 
p \left(  \bm{u}; \bm{\theta}_k^{\star\star\star} \right) = \frac{1}{(\prod_{i=1}^{p}u_{i})|\bm{ \Sigma}_k^{\star\star\star} |(2 \pi)^{\frac{p}{2}}}   \exp\left[-\frac{1}{2}(\ln\bm{ u}-\bm{\mu}_k^{\star\star\star})^{'}\bm{\Sigma}_k^{{\star\star\star}_{-1}}(\ln \bm {u}-\bm{\mu}_k^{\star\star\star})\right].
\end{equation} The derivation of \eqref{eqn:logn} can be found in the Appendix \ref{changeVarUni}, and the application of the procedure therein can be followed for other types of non-Gaussian covariates, thus generalizing the CWM model.


\subsection{Zero-Inflated Generalized Cluster-Weighted Model}

Using the zero-inflated Poisson model (ZIP) as a special and most widely used case of zero-inflated models we further extend the generalized cluster-weighted model class to zero-inflated generalized cluster-weighted model (ZI-GCWM). This choice by no means reduces the generality of our approach. 

We begin by noting that the single component ZIP model assumes that the inflated zeros emanate from both a Bernoulli and Poisson random variables while the non-zeros are assumed to come exclusively from the Poisson random variable  \citep[see][]{Lambert}. However, recent research  extends the single component ZIP models to mixture models for heterogeneous count data with excess zeros \citep[see][]{Bermudez+Karlis:2012}. In mixtures of ZIPs, zeros are assumed to come from multiple different Binomial and Poisson random variables. 

Thus, seen in the context of GCWM, considering the ZIP model, we can split the conditional density of the response variable $Y$, i.e., $p(y|\bm{x},\bm{\vartheta}_k)$, into zero and non-zero densities for each group $k$. The conditional probability mass associated with the event $\{y=0\}$ is characterized by $q(y = 0|\bm{x};\bm{\vartheta}_{k})$. For the event $\{y > 0\}$, the response variable $Y$ is conditionally distributed with density $q(y > 0|\bm{x}; \bm{\vartheta}_{k} )$. All this considered, given the conditional density for the ZIP model, \eqref{eq3} can be re-written as
 \begin{align}
 f(\bm x, y; \Phi)= \sum_{k=1}^{K} \tau_k \left[ q(y = 0|\bm{x};\bm{\vartheta}_{k} ) +  q(y > 0|\bm{x} ; \bm{\vartheta}_{k}  ) \right]   p(\bm{t};\bm{\theta}_k^{\star})p(\bm{w};\bm{\theta}_k^{\star\star})p(\bm{u};\bm{\theta}_k^{\star\star\star})
\end{align}
which characterizes the zero Inflated Generalized cluster-weighted model (ZI-GCWM).

Specifically, have the Poisson conditional density  denoted as $ q^P(y|\bm{x}; \lambda_k) $ where $y \in \{0,1,\dots\}$. Additionally, have vector $\xTilda := [\bm{1},\bm{x}]$ to contain the covariates $\bm{x}$ together with a placeholder for the intercept in the GLM and let $\bm{\beta}_k$ be a row coefficient vector.
The link function will be chosen to be log-link such that
 \begin{align}\label{g1link}
\lambda_k = e^{\xTilda \bm{\beta}_k'} && \text{and} & & 
q^P(y|\bm{ x} ; \lambda_{k} ) = e^{-\lambda_k} \frac{{\lambda_k}^y}{y!}.
 \end{align}
Also, have a Bernoulli model for the conditional density denoted as $ q^{B}(y|\bm{x}; \bm{\bar{\beta}}_k) $, where $\bar{\bm{\beta}_k}$ is a row coefficient vector.  Here, the link function is chosen to be logit link function so that
 \begin{align}\label{g2link}
 \psi_k =  \frac{e^{\xTilda \bm{\bar{\beta}}_k'}}{1+ e^{\xTilda  \bm{\bar{\beta}}_k'}}  && \text{and} && 
 q^B(y | \bm{x} ; {\psi}_k) = \begin{cases}
      \quad \psi_k, & y = 0,\\
     1 -  \psi_k,  & y > 0.
   \end{cases}
 \end{align}
By creating composition of two preceding models, we have the ZIP model in which zero counts come from two random variables. One is the Bernoulli random variable, which generates structural zeros, and the other is the Poisson random variable. The coefficients $\bm{\vartheta}_{k}=\{ \bm{\beta}_{k},  \bm{\bar{\beta}}_k \}$ correspond to the two above introduced conditional densities where the coefficients can be estimated as in \cite{Lambert}. The $k$th component of ZIP conditional density $q(y|\bm{x}; \bm{\vartheta}_{k}  )$ is 
 \begin{align*}
 q( y = 0| \bm{x} ; \bm{ \vartheta}_{k}  ) = \psi_k + (1 - \psi_k)e^{-\lambda_k}  & &  \text{and}  & &
q(y > 0 |  \bm{x} ; \bm{ \vartheta}_{k}  ) = (1 - \psi_k)e^{-\lambda_k} \frac{\left(\lambda_k \right)^y  }{y!}.
 \end{align*}
The parameter $\psi_k$ denotes the mean of the Bernoulli distribution of the $k$th component from which extra zeros emanate, and the parameter $ \lambda_k $ characterizes the $k$th Poisson distribution. 

In our numerical example related to automobile insurance, it will be shown that this allows for a more nuanced approach to handling the inflation of zeros coming from heterogeneous sources. 

\section{Parameter Estimation}\label{sec:estmeth}

The common approach for estimating parameters in finite mixture models is based on the EM algorithm \citep[see][for examples]{mcnicholas16a}.
The estimation of the developed Bernoulli-Poisson partitioning method is split into two EM algorithms. The first EM algorithm partitions the sample space, while the second EM algorithm optimizes the zero inflated portion.

\subsection{Note on Bernoulli-Poisson Sample Space Partitioning}

Unfortunately, in the context of mixture models for heterogeneous count data with excess zeros the difficulties are apparent  during the maximization step of the EM algorithm when means of covariates are very close together \cite[see][]{LimHwa}. The good news is that the misclassification error can be reduced using parsimonious models for the independent variables as in  \cite{McNicholas:2010}. 

However, in this work, we propose a new method to rectify this problem and partition the dataset using Bernoulli and Poisson GCWMs. Furthermore, we construct a ZI-GCWM using the previously generated Bernoulli and Poisson GCWMs. In the first EM algorithm we estimate parameters pertaining to the GCWM under the assumption of a Poisson model and, separately, we carry out the same process under the assumption of a Bernoulli model. Using the obtained parameter estimates from the two separate applications of the EM algorithm, we set the initialization parameters for the second EM algorithm pertaining to parameter estimation of the ZI-GCWM. The work of \cite{Lambert} specifies that the MLE estimates for the separate Poisson and Bernoulli models provide an excellent initial guess, allowing EM to converge quickly for ZIPs. The Bernoulli-Poisson sample space partitioning method consists of two separate EM algorithms. The first EM algorithm is for generating the GCWM models, while the second EM is for optimizing the ZI-GCWM.

Here, the joint probability density function $f^{ZI}$ becomes
 $$f^{ZI}(\bm{x},y,\Phi) = \sum_{k=1}^{K} \tau_k q^{ZI}_{k}(y|\bm{x};  \bm{\bar{\beta}}_k,\bm{ \beta}_k)  p(\bm{t};\bm{\theta}_k^{\star})p(\bm{w};\bm{\theta}_k^{\star\star})p(\bm{u};\bm{\theta}_k^{\star\star\star}). $$  
 The new conditional density is now result of a model in which each component is captured by the conditional probability density function that is a mixture of particular Bernoulli and particular Poisson densities
\begin{align}
q^{ZI}_{k}(y|\bm{x};  \bm{\bar{\beta}}_k,\bm{ \beta}_k) & := q^B(y|\bm{x}; \bm{\bar{\beta}}_k) +(1-  q^B(y|\bm{x}; \bm{\bar{\beta}}_k) ) q^P(y|\bm{x};\bm{\beta}_k) \nonumber \\
& = q(y = 0|\bm{x};\bm{\vartheta}_{k} ) +  q(y > 0|\bm{x} ; \bm{\vartheta}_{k}), \quad k \in \{ 1, ..., K  \}.
\label{ziGCWM}
\end{align}

The initialization parameters for the second EM algorithm are provided by Bernoulli and Poisson GCWMs from \eqref{ziGCWM} giving parameter pairs ($ \psi_k,\lambda_k  $). The second EM procedure then optimizes the zero inflated GCWM. The ZI-GCWM is compared against the standard Poisson GCWM using a likelihood ratio test which is discussed in Section~\ref{subsec:: compareZero}.

\subsection{EM Algorithm for Partitioning of Sample Space}

The EM algorithm is based on the local  maximum likelihood estimation. 
The initial values of the parameter estimates can be generated from a variety of strategies outlined in \cite{initialPaperGrassiaRef}. 
 The algorithm proceeds by alternation of the E- and M-steps to update parameter estimates. 
The convergence criterion of the EM algorithm is based on the Aitken acceleration. It is used to estimate the asymptotic maximum of the log-likelihood at each iteration of the EM algorithm when the relative increase in the log-likelihood function is no bigger than a small pre-specified tolerance value or the number of iterations reach a limit. 
To find an optimal number of components, maximum likelihood estimation is obtained over a range of $K$ groups, and the best model is selected based on the Bayesian information criterion (BIC).   

In this subsection, we explain the parameter estimation in line with the GCWM methodology proposed by \cite{Ingrassia+Punzo+Vittadini+Minotti:2015}. The proposed GCWM  is based on the assumption that $q(y|\bm{x},\bm{\vartheta}_k)$ belongs to the exponential family of distributions that are strictly related to GLMs. The link function in each group $k$ defines the relationship between the linear predictor and the expected value of the distribution function.   
Here we are interested in the estimation of the vector $\bm {\beta}_k$, thus the distribution of $Y|\bm{x}$ is denoted by $q(y|\bm{x}; \bm{\beta}_k, \nu_k)$, where $\nu_k$ signifies an additional parameter to account for when a distribution belongs to a two-parameter exponential family.\footnote{In the work of \cite{Ingrassia+Punzo+Vittadini+Minotti:2015} this parameter is referred to as $\lambda_k$.} 

Recall that the marginal distribution $p(\bm{x}; \bm \theta_k)$ has the following components: $p(\bm{t}; \bm \theta_k^{\star})$, $p(\bm{w}; \bm \theta_k^{\star\star})$, and $p(\bm{u};\bm \theta_{k}^{\star\star\star})$. The first marginal density  $p(\bm{t}; \bm \theta_k^{\star}:=( \bm {\mu}_k^{\star}, \bm{\Sigma}_k^{\star}) )$ is modeled as a  Gaussian distribution with mean $\bm {\mu}_k^{\star}$ and covariance matrix $\bm{\Sigma}_k^{\star}$. 
 The marginal density of discrete covaraites $p(\bm{w};\bm{\theta}_{k}^{\star\star})$ is assumed to have for each finite discrete covariate in $\bm{W}$, a representative binary vector $\bm{w}^r=(w^{r1},\ldots,w^{rc_r})^{'}$, where $w^{rs}=1$ if $w_r = s\in\{1, \ldots, c_r\}$, 
and $w^{rs}=0$ otherwise.

Given the preceding assumptions about discrete covariates, the marginal density is written as
\begin{align}
p(\bm {w}; \bm {\gamma_k})=\prod_{r=1}^{d}\prod_{s=1}^{c_r}(\gamma_{krs} )^{w^{rs}}
\label{eq31}
\end{align}
for $k=1, \ldots, K$, where $\bm {\gamma}_k=(\gamma_{k1}^{'}, \ldots, \gamma_{kd}^{'})^{'}$, $\bm \gamma_{kr}=(\gamma_{kr1}^{'}, \ldots, \gamma_{krc_d}^{'})^{'}$, $\gamma_{krs} > 0$, and  $\sum_{s=1}^{c_r}\gamma_{krs}$, $r=1,\ldots,q$. The density $p(\bm {w}, \bm{\gamma}_k)$ represents the product of $d$ conditionally independent multinomial distributions with parameters $\bm{\gamma}_{kr}$, $r=1,\ldots, d$. Finally, the third marginal density $p(\bm{u};\bm{\theta}_{k}^{\star\star\star})$ will be modelled with a multivariate log-normal distribution having a location parameter vector $ \bm{\mu}_k^{\star\star\star}$ and scale parameter matrix $\bm{\Sigma}_k^{\star\star\star} $.

Let $(\bm x_1, y_1),\ldots, (\bm x_n, y_n)$ be a sample of $n$ independent observations drawn from model in \eqref{eq3}. Consider a latent random variable $Z_{ik}$.  The realization $z_{ik}$ of the latent indicator variable takes the value of $z_{ik}=1$ indicating that observation $(\bm{x_i}, y_i)$ originated from the $k$th mixture component and $z_{ik}=0$ otherwise.

 Given the sample, the complete-data likelihood function $L_c(\bm\Phi)$ is given by
\begin{align}
L_c(\bm\Phi)=\prod_{i=1}^{n}\prod_{k=1}^{K}\left[{\tau_k}q(y_i|x_i; \bm \beta_k, \nu_{k})p(t_i; \bm\mu_k^{\star}, \bm\Sigma_k^{\star}) p(w_i; \gamma_k)p(u_i; \bm{\mu}_k^{\star\star\star},\bm{\Sigma}_k^{\star\star\star}) \right]^{z_{ik}},
\label{eq27}
\end{align}

Taking the logarithm of \eqref{eq27}, the complete-data log-likelihood is 
\begin{align}
\ell_c(\bm\Phi)= \sum_{i=1}^{n}\sum_{k=1}^{K}{z_{ik}}\big[&\log(\tau_{k}) + \log{q}(y_i|x_i; \bm{\beta}_k,\nu_k)\nonumber\\&+  \log p(t_i; \bm{\mu}_k^{\star}, \bm{\Sigma}_k^{\star}) + \log p(w_i; \bm{\gamma}_k) +\log {p}(u_i; \bm{\mu}_k^{\star\star\star},\bm{\Sigma}_k^{\star\star\star}) \big].\label{CompleteLiklihood}
\end{align}

On the $(s+1)$th iteration, the E-step requires calculation of the conditional expectation of $\ell_c(\bm\Phi)$. Because $\ell_c(\bm\Phi)$ is linear with respect to  $z_{ik}$, we simplify the calculation to the current expectation of $Z_{ik}$, where $Z_{ik}$ is the random variable corresponding to the realization $z_{ik}$. Given the previous parameters $\bm\Phi^{(s)}$ and the observed data,  we calculate the current conditional expectation of $Z_{ik}$ as
\begin{equation*}\begin{split}
    {\pi_{ik}}^{(s)} &= {E}[Z_{ik} |(\bm{x_i}, y_i); \bm{\Phi}^{(s)}]\\
     &= \frac{{\tau_k}^{(s)}q(y_i|x_i; \bm \beta_k^{(s)}, \nu^{(s)}_{k})p(t_i; \bm\mu_k^{{\star}(s)}, \bm\Sigma_k^{{\star}(s)}) p(w_i; \bm \gamma_k^{(s)})p(u_i; \bm{\mu}_k^{\star\star\star (s)},\bm{\Sigma}_k^{\star\star\star (s)})}{f(\bm{x}_i, y_i; \bm{\Phi}^{(s)})
\label{eq29}                       }.
\end{split}\end{equation*}
%
On the M-step of the $(s+1)$th iteration, the conditional expectation of $\ell_c(\bm\Phi)$ denoted as a function $Q(\Phi|\Phi^{(s)})$ is maximized with respect to $\Phi $, where the values of $z_{ik}$ in \eqref{CompleteLiklihood} are replaced by their current expectations $\pi_{ik}$ yielding 
\begin{equation}\begin{split}
Q(&\bm\Phi|\bm\Phi^{(s)}) = \sum_{i=1}^{n}\sum_{k=1}^{K}{\pi_{ik}^{(s)}} \big[\log(\tau_{k}) + \log{q}(y_i|x_i;\bm{\beta}_k,\nu_k)+ \log p(t_i; \bm{\mu}_k^{\star}, \bm{\Sigma}_k^{\star})  + \log p(w_i; \bm{\gamma}_k)\\ 
&\qquad\qquad\qquad\qquad+\log {p}(u_i; \bm{\mu}_k^{\star\star\star },\bm{\Sigma}_k^{\star\star\star })\big] \\
&=\sum_{i=1}^{n}\sum_{k=1}^{K}{\pi_{ik}^{(s)} \log(\tau_{k}) + \sum_{i=1}^{n}\sum_{k=1}^{K}{\pi_{ik}^{(s)}}\log{q}(y_i|x_i;\bm{\beta}_k},\nu_k) +\sum_{i=1}^{n}\sum_{k=1}^{K} {\pi_{ik}^{(s)}}\log p(t_i; \bm{\mu}_k^{\star}, \bm{\Sigma}_k) \\
&\qquad\qquad\qquad\qquad+\sum_{i=1}^{n}\sum_{k=1}^{K}{\pi_{ik}^{(s)}}\log p(w_i; \bm{\gamma}_k) + \sum_{i=1}^{n}\sum_{k=1}^{K}{\pi_{ik}^{(s)}}\log {p}(u_i; \bm{\mu}_k^{\star\star\star},\bm{\Sigma}_k^{\star\star\star}).\label{Qfunction}
\end{split}\end{equation}

The M-step requires maximization of the $Q$-function with respect to $\bm \Phi$ which can be done separately for each term on the right hand side in \eqref{Qfunction}. 
As a result, the parameter updates on the $(s+1)$th iteration are
\begin{align*}
{\hat{\tau}_k}^{(s+1)}&=\frac{1}{n} \sum_{i=1}^n \pi_{ik}^{(s)}, && && {\hat{\bm{\mu}}_k}^{\star (s+1)}=\frac{1}{\sum_{i=1}^n \pi_{ik}^{(s)}} \sum_{i=1}^n \pi_{ik}^{(s)}\bm t_i, &&  && {\hat{\bm \gamma}^{(s+1)}_{kr}} =\frac{\sum_{i=1}^n \pi_{ik}^{(s)} \omega^{rs}_i} {\sum_{i=1}^n \pi_{ik}^{(s)}},
\end{align*}
$$
 {\widehat{\bm \Sigma^{}}_k}^{\star(s+1)}=\frac{1}{\sum_{i=1}^n \pi_{ik}^{(s)}} \sum_{i=1}^n \pi_{ik}^{(s)}(\bm t_i-\hat{\bm \mu}^{(s+1)}_k) (\bm t_i-\hat{\bm \mu}^{(s+1)}_k)^{'}.
$$
Parameter updates for the log-normal distribution are as follows
\begin{equation*}\begin{split}
{\hat{\bm \mu}_k}^{\star\star\star (s+1)}&=\frac{1}{\sum_{i=1}^n \pi_{ik}^{(s)}} \sum_{i=1}^n \pi_{ik}^{(s)}\ln \bm u_i,\\
{\widehat{\bm \Sigma}_k}^{\star\star\star(s+1)}&=\frac{1}{\sum_{i=1}^n \pi_{ik}^{(s)}} \sum_{i=1}^n \pi_{ik}^{(s)}(\ln \bm u_i-\hat{\bm \mu}^{\star\star\star(s+1)}_k) (\ln \bm u_i-\hat{\bm \mu}^{\star\star\star(s+1)}_k)^{'}. 
\end{split}\end{equation*}
For each $k=1,\ldots,K$, the update for $\bm{\beta}_k$ could be computed by maximizing
\begin{align}
\sum_{i=1}^{n}\pi^{(s)}_{ik} \log{q}(y_i|\bm x_i;\bm \beta_k,\nu_k).
\label{eq30}
\end{align}
The numerical optimization for each term is discussed in \cite{Wedel+DeSabro:1995} and \cite{Wedel:2002}.
For additional implementation information, the reader is referred to the manual of the {\tt flexCWM} package \citep{Ingrassia+Punzo+Vittadini+Minotti:2015} for ${\sf R}$ \citep{R18}.

For modelling severity, each observation $y_i$ must be weighted according to the number of claims the client has incurred \citep[see][pages 118--119]{frees2015}. Thus \eqref{eq30} is re-written as 
\begin{align}
\sum_{i=1}^{n}\pi^{(s)}_{ik} \mathcal{\omega}_i \log{q}(y_i|\bm x_i;\bm \beta_k,\nu_k),
\label{eqFrees}
\end{align}
which is maximized to give the update for $\bm{\beta}_k$.
Here, $\mathcal{\omega}_i$ is the number of claims the client occurs over the exposure period. Because $\mathcal{\omega}_i$ is constant at every EM iteration $s$, the flexCWM package is easily amenable to account for this methodological adjustment.


\subsection{EM Algorithm for Zero-Inflated Model} 
For a zero-inflated model, the EM-algorithm follows a similar procedure as above to optimize the conditional density given in \eqref{ziGCWM}.  Specifically, the log-likelihood function of $\psi_k$ and $\lambda_k$ is 
\begin{equation*}\begin{split}
l(\psi_k,\lambda_k| \{y_i\}_{i=1}^n,\{\bm{x}_i\}_{i=1}^n) &= \sum_{\{y_i = 0\}} \log \big[ e^{ \bm{ \xTilda}_i \bm{\bar{\beta}}_k^{'}  } + \exp{( - e^ { -\bm{\xTilda}_i \bm{\beta}_k^{'} })} \big]  \\ & +  \sum_{\{y_i > 0\}} \left( y_i \xTilda_i \bm{\beta}_k^{'} + e^{ \xTilda_i \bm{\beta}_k^{'} } \right)  - \sum_{i=1}^n  \log \left(1 + e^ {\xTilda_i \bm{\bar{\beta}}_k^{'} } \right) - \sum_{\{y_i > 0\}} \log(y_i ! ).
\end{split}\end{equation*}
Due to the first term, the log-likelihood function is difficult to maximize, however \cite{Lambert} gives a meaningful solution.  Consider a random variable $\zZ_{ik}$ indicating with ${\zz_{ik}} = 1$ when $y_i$ is generated from the Bernoulli random variable of partition $k$, and $\zz_{ik} = 0$ when $y_i$ is generated from the Poisson random variable of the same partition.  Then, the complete-data log-likelihood is 
\begin{align*}
l_c(\psi_k,\lambda_k| \{y_i\}_{i=1}^n,\{\bm{x}_i\}_{i=1}^n,{\bm{\zz}_k}) &= \sum_{i=1}^n \left( \zz_{ik}\xTilda_i \bar{\bm{\beta}_k }^{'} - \log\left(1+ e^{ \xTilda_i \bar{\bm{\beta}_k }^{'}}\right) \right)  \\ & + \sum_{i=1}^n (1-\zz_{ik}) (y_i \xTilda_i \bm{\beta}_k^{'}  - e^{\xTilda_i \bm{\beta}_k^{'}})+ \sum_{i=1}^n (1-\zz_{ik})\log(y_i!)\\
&= l_c(\psi_k;\{y_i\}_{i=1}^n,\{\bm{x}_i\}_{i=1}^n,{{\bm{\zz}_k}}) + l_c(\lambda_k; \{y_i\}_{i=1}^n,\{\bm{x}_i\}_{i=1}^n,{{\bm{\zz}_k}}) \\
&+ \sum_{i=1}^n (1- \zz_{ik})\log(y_i!), 
\end{align*}
where $\bm{\zz}_k := \left[\zz_{1k}, ..., \zz_{nk} \right]$ is a realization of $\bm{\zZ}_k  := \left[\zZ_{1k}, ..., \zZ_{nk} \right]$.
 Note that $l_c(\psi_k,\lambda_k|\{y_i\}_{i=1}^n,\{\bm{x}_i\}_{i=1}^n,\bm{\zz}_k)$ can be separated allowing the maximization of $l_c(\psi_k; \{y_i\}_{i=1}^n,\{\bm{x}_i\}_{i=1}^n,\bm{\zz}_k)$ and $l_c(\lambda_k; \{y_i\}_{i=1}^n,\{\bm{x}_i\}_{i=1}^n,\bm{\zz}_k) $ independently for parameters $\psi_k$ and $\lambda_k$. With the EM algorithm, maximization of parameters are performed iteratively between estimating $\zZ_{ik}$ with its expectation under current estimates for $\lambda_k$ and $\psi_k$ (E-Step), and then maximizing the conditional expectation of the complete-data log-likelihood (M-Step). 


In the E-step, using current estimates $\psi_k^{(s)}$ and $ \lambda_k^{(s)} $ 
we calculate the expected value of ${{\zZ}_{ik}}$ by its posterior mean ${\hat{\zz}_{ik}^{(s)}}$ for each cluster $k$ at iteration $s$ as
\begin{align*}
{\hat{\zz}}_{ik}^{(s)} = \begin{cases}  \left[ 1 + \exp{\big(-\xTilda_i \bar{\bm{\beta}_k}^{'(s)} - e^ {\bm{\xTilda_i} \bm{\beta}_k^{'(s)}} \big) } \right]^{-1}, &  y_{i} = 0 \\
  0 \quad , & y_{i}> 0 .
\end{cases}
\end{align*}
%
The M-Step can be split into the maximization of two complete data log-likelihoods and the $\hat{\bm{\zz}}_k$ calculated from the previous iteration $(s)$ as
\begin{align}
 l_c(\lambda_k; \{y_i\}_{i=1}^n,\{\bm{x}_i\}_{i=1}^n| \hat{\bm{\zz}}_k^{(s)}) &= \sum_{i=1}^n (1- \hat{\zz}_{ik}^{(s)}) (y_i \xTilda_i \bm{\beta}_k^{'}  - e^{\xTilda_i \bm{\beta}_k^{'}})\label{eq7}.\\
l_c(\psi_k;\{y_i\}_{i=1}^n,\{\bm{x}_i\}_{i=1}^n|\hat{\bm{\zz}}_k^{(s)}) &=\sum_{i=1}^n \left( \hat{\zz}_{ik}^{(s)} \xTilda_i \bar{\bm{\beta}_k }^{'} - \log \left(1+ e^{ \xTilda_i \bar{\bm{\beta}_k }^{'}} \right) \right). \label{eq6}   
 \end{align}

The maximization of \eqref{eq7} for GLM coefficients $\lambda_k$ can be carried out using a weighted log-linear Poisson regression with weights $1 - \hat{\zz}_{ik}^{(s)}$ \citep[see][]{McCullaghNelder1989}), yielding $\lambda_k^{(s+1)}$.
While the parameter $\psi_k$ for \eqref{eq6} can be maximized over a gradient yielding $\psi_k^{(s+1)}$ \citep[see][]{Lambert}. 

\subsection{Comparing Zero-Inflated Models}\label{subsec:: compareZero}

Until recently the Vuong Test for non-nested models  \citep{vuongTest} has frequently been used to compare zero-inflated model with their non-zero inflated counterpart. However \cite{misuse} shows that a zero-inflated model and its non-zero inflated counterpart do not satisfy Vuong's criteria for non-nested models, and hence such use of the test is incorrect. Furthermore, the Vuong Test fails to identify evidence of zero-deflation leading to inconsistencies in the hypothesis test \citep[see][]{misuse}. To rectify this, \cite{newIntuitive} show that it is sufficient to test for zero-modificiation in the form of a likelihood ratio test, where the hypotheses are
\begin{align*}
& & H_0: \psi_k = 0 \quad \text{vs.} \quad H_1: \psi_k > 0, & &
\end{align*}
and the test statistic $\varphi$ is given by
\begin{equation}
\varphi = -2 \big[l(\tilde{\lambda_k}; \{y_i\}_{i=1}^n,\{\bm{x}_i\}_{i=1}^n) - l(\lambda_k, \psi_k; \{y_i\}_{i=1}^n,\{\bm{x}_i\}_{i=1}^n )\big].
\label{LRTest}
\end{equation}
The test statistic \eqref{LRTest} is shown to follow a chi-squared distribution ($0.5\chi^2_0+0.5\chi^2_{m}$) where $m$ is the  number of degrees of freedom of the model. For our purposes, $m$ is the number of covariates selected for the Bernoulli model in (\ref{g2link}) --- see Likelihood Ratio Tests in \cite{newIntuitive}. Note that  a $\chi^2_0$ distribution has 0 degrees of freedom, hence the $95{\mbox{\sl{th}}}$ percentile of a  $0.5\chi^2_0+0.5\chi^2_{m}$ distribution is equal to the $90{\mbox{\sl{th}}}$ percentile of a $\chi^2_{m}$ distribution and the test statistic at a level of significance of $\alpha=0.05$ equals that of a one-tailed $\chi^2$ test, with $m$ degrees of freedom, at a level of significance of $\alpha=0.10$ \citep[see][]{newIntuitive}. 

The function 
$l(\tilde{\lambda_k}; \{y_i\}_{i=1}^n,\{\bm{x}_i\}_{i=1}^n)$ is the log-likelihood of a single component GCWM Poisson model parametrized by $\tilde{\lambda_k}$. Recall that $\psi_k$ is the zero-inflation parameter of the $k$th partition.  In our approach, we will be using \eqref{LRTest} to test for evidence of zero-inflation for $k$th group, and then using BIC for model comparisons on $k$th group. This approach quickly determines if there is zero-inflation.  When evidence of zero-inflation is established, we search for the best linear model using the BIC. 

\section{Numerical Application}\label{sec:numapp}
\subsection{Dataset}
We illustrate the use of ZI-GCWM on the French motor severity and frequency datasets which are available as part of the {\sf R} package {\tt CASdatasets} \citep{Dutang+Charpentier:2016}. Previously they were used by \cite{Charpentier:2014} who demonstrated various GLM modeling approaches for fitting frequency and severity. 
The dataset consists of 413,169 motor third-party liability policies with the associated risk characteristics. The loss amounts by policy ID are also provided. 
\begin{table}[!htb]
\begin{center}
    \caption{The description of variables in the French Motor Third-Part Liability dataset.}
      \centering
       {\small \begin{tabular}{ll}
\hline
Attribute & Description \\
\hline
Policy ID & Unique identifier of the policy holder\\
Claim Nb & Number of claims during exposure period  (0,1,2,3,4)\\
Exposure & The exposure of policy in years (0--1.5) \\
Power & Power level of car ordered categorical (12 levels )\\
Car Age & Car age in years \\
Driver Age & Age of a legal driver \\
Brand & Car brands (7 types) \\
Gas & Diesel or Regular \\
Region & Regions in France (10 classifications)\\
Density & Number of inhabitants per km$^2$ \\
Loss Amount & Portion of claim the insurance policy pays\\
Severity & Average Claim calculated from aggregating Loss Amount and dividing by Claim Nb \\
\hline\end{tabular}}
\end{center}
\end{table}

\subsection{Discussion and Results}
\subsubsection{Modelling Severity}
In this section, we show improved results for GCWM  over CWM in modeling French motor losses. Furthermore, we investigate the results of the GCWM model for the valuation of heterogeneous risk. 

In the numerical analysis we consider the following covariates: population density ($Density$), driver age ($DriverAge$), car age ($CarAge$), car power level ($Power$),  and geographical region in France ($Region$). 
The
$CarAge$ is modelled as a categorical variable with five categories: $[0,1)$, $[1,5)$, $[5,10)$, $[10,15)$, and $15+$. Additionally, $DriverAge$ is modelled as a categorical variable with five categories: $[18,23)$, $[23,27)$, $[27,43)$, $[43,75)$, and $75+$. $Power$ is modelled into three categories as in \cite{Charpentier:2014}:
DEF, GH, and other. The fitted model is defined with the following expression
\begin{align}
g(\mathbb{E}\left[Y_{Severity}|\bm{x},  \bm{\beta}_k  \right]) = 
 & \beta_{k0} +  \beta_{kDensity}x_{Density}+ \beta_{kCar Age} x_{Car Age}+ \beta_{kDriver Age} x_{Driver Age} + \nonumber \\ &  \beta_{kRegion} x_{Region} + \beta_{kPower} x_{Power} \label{regressionModel}  =: \bm{\xTilda} \bm{\beta}_k^{'}.
\end{align}
The canonnical log-link function $g$ is used for the GCWM in \eqref{regressionModel}. 

Beginning with the continuous covariate $Density$, we want to inspect the shape of its univariate data to see if it follows Gaussian distribution. 
The left-hand side of Figure \ref{fig:vet1} clearly reveals that the the $Density$ is rather skewed right with several observations that report high value of density. This indicates a need for a transformation. The log-normal transformation clearly improves the fit (see the right side of Figure \ref{fig:vet1}).
\begin{figure}[!htb]
\begin{center}
\includegraphics[scale=0.63]{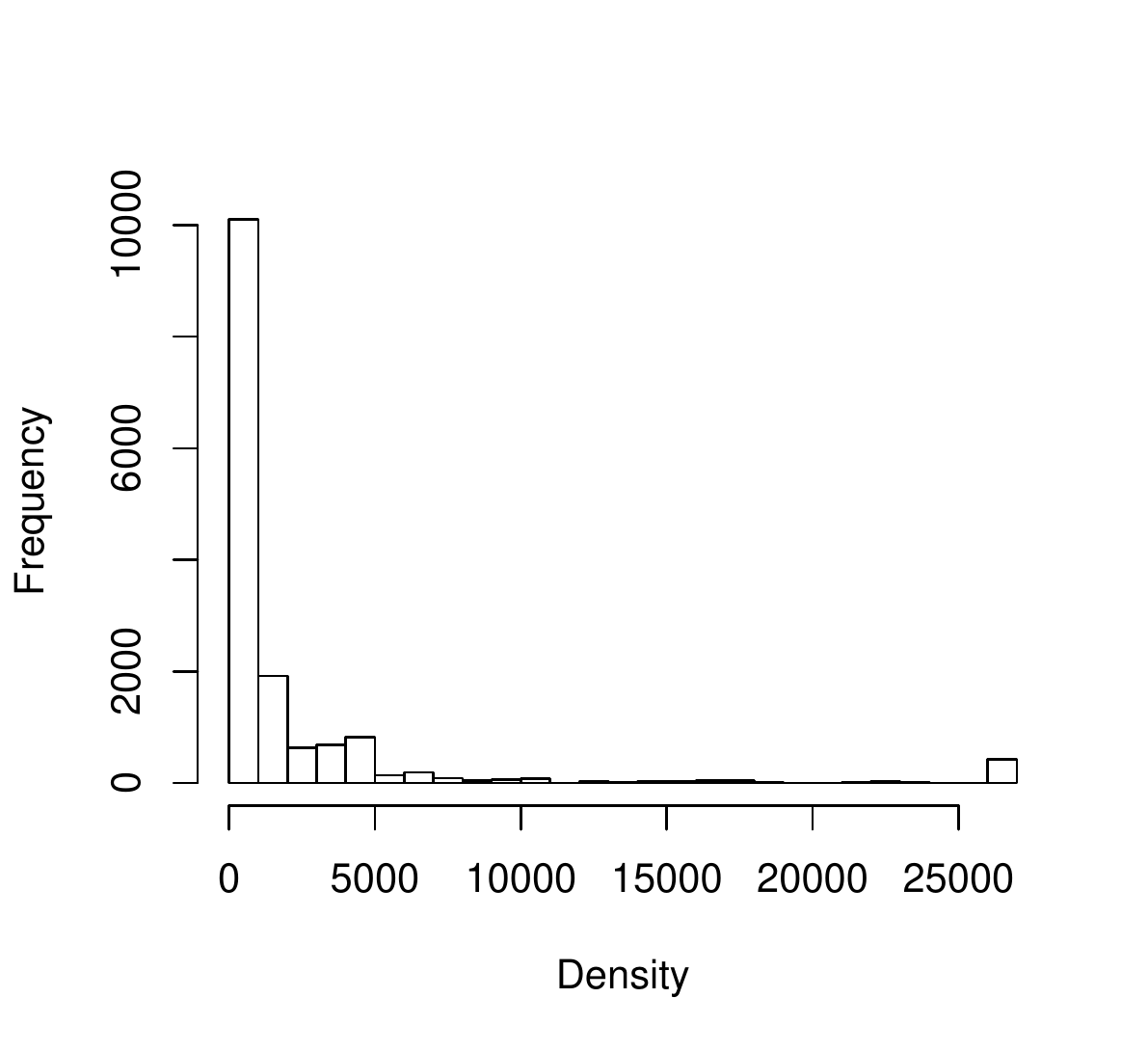}
\includegraphics[scale=0.63]{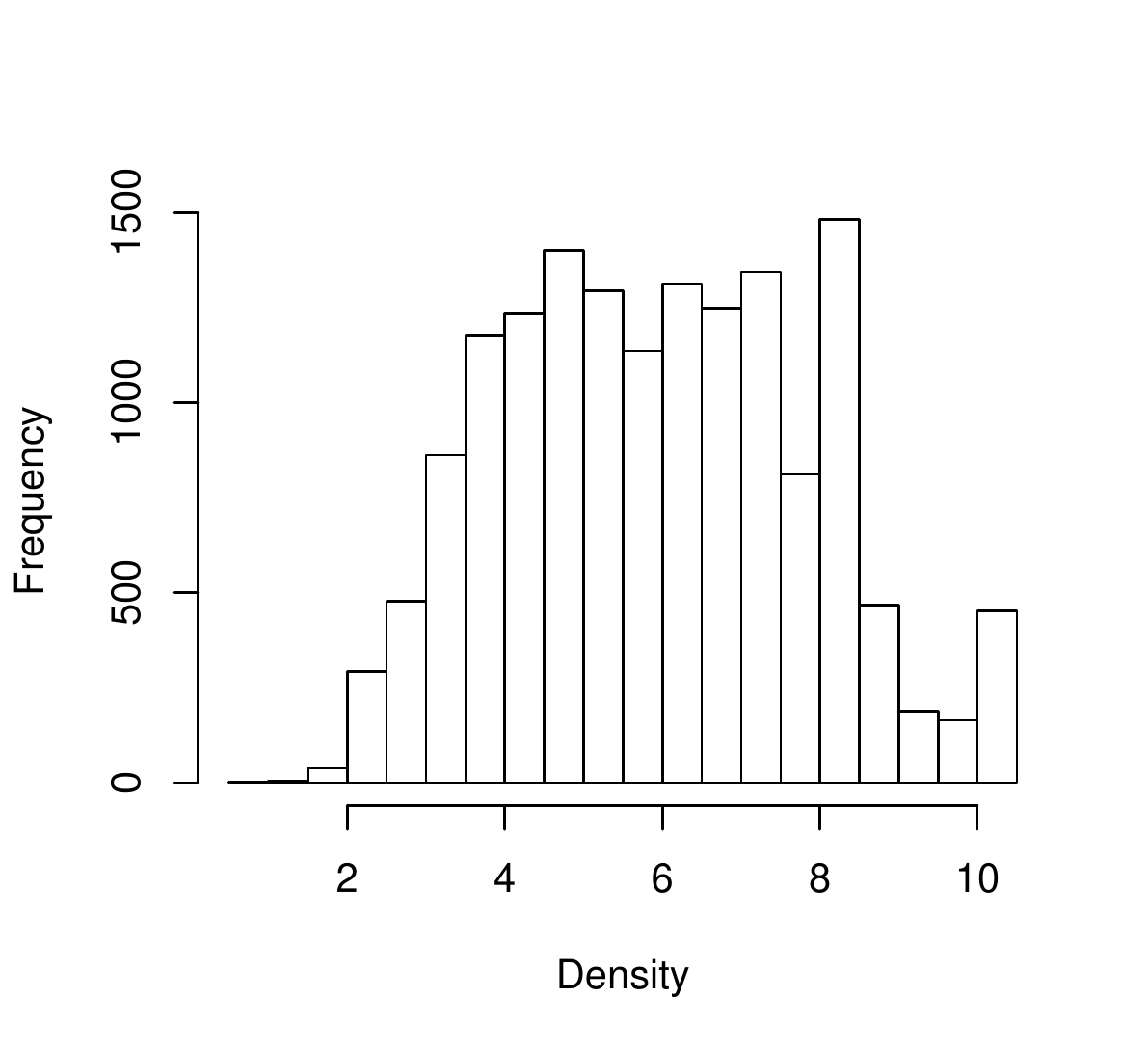}
\end{center}
\vspace{-0.2in}
\caption{Density variable: Left figure shows the fit when Gaussian distribution is imposed (CMW approach) to highly skewed data. Right figure shows the fit when log-normal assumption is applied (GCWM approach).}
\label{fig:vet1}
\end{figure}

Given the log-normal assumption, the result of the transformation is reflected in a better AIC and BIC for GCWM over CWM. Table~\ref{comparingCWM_models} shows a considerable difference in BIC and AIC comparing CWM and GCWM. The five component CWM with a BIC of $281,680$ is significantly higher than the four component GCWM with a considerably lower BIC of $88,564$.
\begin{table}[!htbp] \centering
  \caption{Comparison of AIC and BIC for CWM versus GCWM for $K$ number of clusters estimated. }\label{comparingCWM_models}
\begin{tabular}{@{\extracolsep{5pt}} rrrr}
\hline
Model & $$ & AIC & BIC \\
\hline
CWM & $1$ & $351,965$ & $352,149$\\
& $2$ &  $314,039$  & $314,414$ \\
& $3$ & $300,503$ & $301,068$ \\
& $4$ &  $285,979$ & $286,735$ \\
& $\bm{5}$ & $\bm{280,732}$ & $\bm{281,680}$ \\
\hline
GCWM &  1 & $110,627$ &  $110,810$ \\
& $2$ &  $88,828$ & $89,203$ \\
& $3$  & $88,338$  &$ 88,903$ \\
& $\bm{4}$ &  $ \bm{87,808} $ & $ \bm{88,564} $ \\
& 5 & $88,009$  & $88,956$ \\
\hline
\end{tabular}
\end{table}

	We now investigate the results of GCWM in relation to the valuation of risk. For practical uses, finding clusters allows us to create different classifications of risk for various groups of drivers.
After fitting the model, we then inspect the size of each cluster. The GCWM approach has chosen four components as the best model to represent the data. The size of each cluster is displayed in Table~\ref{table:sizeSev}. Attention is brought to largest quantity of drivers that are grouped into (blue) Cluster 3. This accounts for $ 47 \% $ of all drivers and is fairly concentrated in the center of Figure \ref{fig:vet1a}. From the results we can create an insurance model with the distinct characteristics.
\begin{table}[!htb]
\centering
\caption{Size and colours of clusters for the GCWM a model.}
\label{table:sizeSev}
\begin{tabular}{rrrr}
\hline
Cluster 1   & Cluster 2  & Cluster 3   & Cluster 4    \\
\hline
$3,064$ & $1,873$  &$ 7,259$ & $3,194$ \\
red & green & blue & orange \\
\hline
\end{tabular}
\end{table}
\begin{figure}[!htb]
\begin{center}
\includegraphics[scale=0.675]{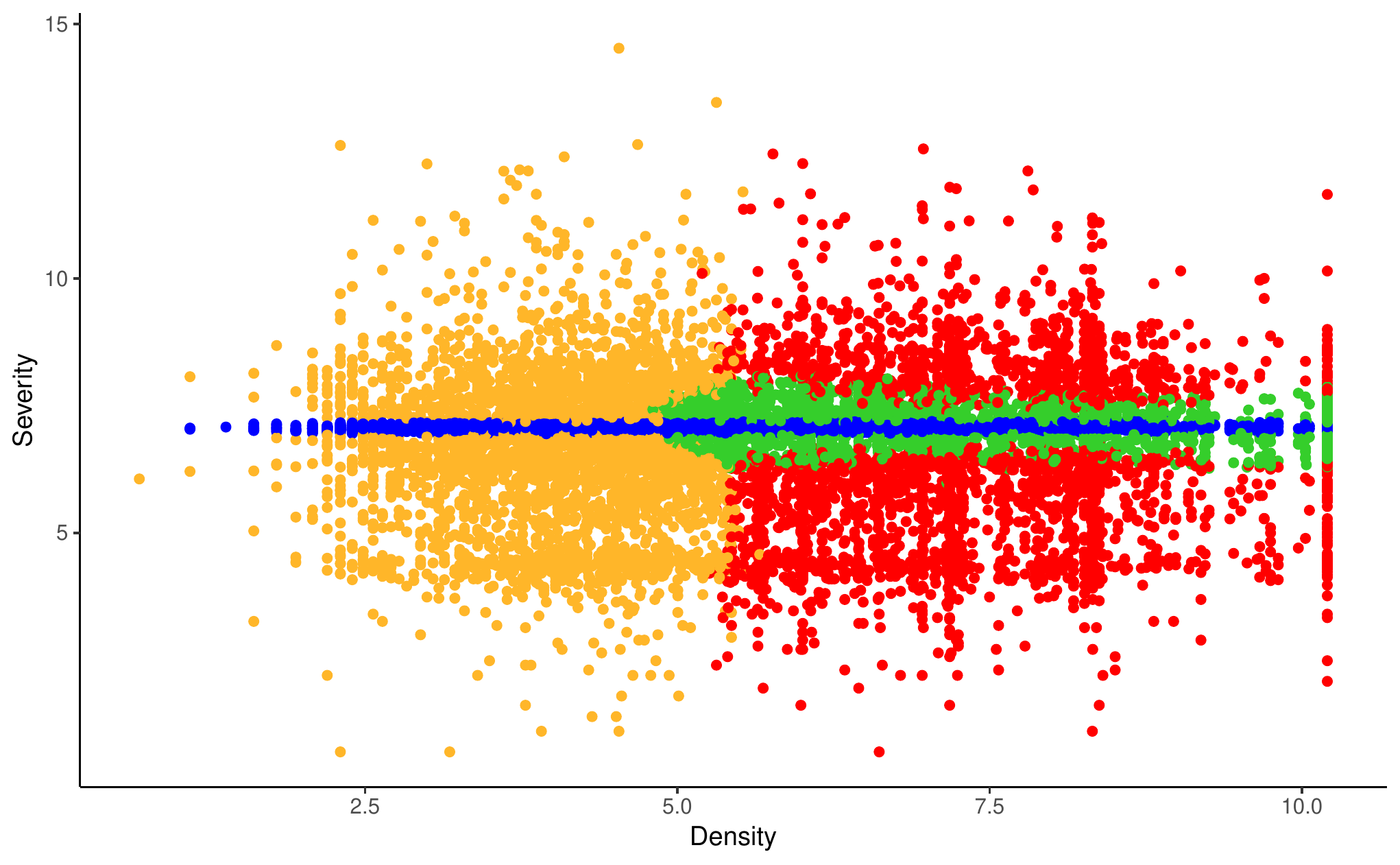}
\end{center}
\vspace{-0.2in}
\caption{Showing clusters by color scheme: Cluster 1 - red, Cluster 2 - green, Cluster 3 - blue, Cluster 4 - orange for $Severity$ vs $Density$ on a log scale. }
\label{fig:vet1a}
\end{figure}

The (blue) Cluster 3 drivers have both low variability and low level of claims, thus can be insured with a lower rate than all other drivers. Similarly, the (green) Cluster 2 drivers, have also low variability of claims but a higher level of claims, thus they should have a rate higher than (blue) Cluster 3  drivers. The (red) Cluster 1 drivers have the next highest level of claims and variability of the clusters in Figure \ref{fig:vet1}. Finally, the (orange) Cluster 4  drivers have the highest level of claims and variability in claims out of any of the other clusters. From a risk management perspective, the drivers belonging to  this cluster should be insured at the highest rate.

Table \ref{table:volSev} shows a breakdown of the types of drivers, ordered by volatility in descending order. Beginning with V1 volatility level, these drivers tend to have claims between \euro$1,033$ and  \euro$1,325$, with a standard deviation of  \euro$52$, and a mean of  \euro$1,171$. Moving to V2 , these drivers have the second lowest level of volatility. Also, drivers in this volatility level tend to have claims anywhere between  \euro$395$ and  \euro$3,221$, with a standard deviation of  \euro$559$, and a mean of  \euro$1,350$.  Proceeding to V3 volatility level, we observe that volatility in claims is greater than the preceding levels. Drivers in this volatility level have claims anywhere between  \euro$2$ and \euro$281,403$, with a mean of  \euro$2,956$, and a standard deviation of  \euro$11,878$. Finally, with V4 denotes the level of highest volatility. A claim in this level reach the highest recorded value of  \euro$2,036,833$. The mean of  claims in this volitility level is \euro$3,771$, and a standard deviation of  \euro$40,334$. 
\begin{table}[!htb]
\centering
\caption{Summarized volatility information of each cluster for Claims.}
\label{table:volSev}
\begin{tabular}{l|rrrr}
\hline
Volatility Level     & \multicolumn{4}{c}{Claim size}   \\ 
 (Cluster)     & Minimum & Mean  & Maximum & $\sigma$    \\
\hline
V1 (3) & 1,033 & 1,171 & 1,325 & \textbf{52} \\
V2 (2) & 395 & 1,350 & 3,221 & \textbf{559} \\
V3 (1) & 2 & 2,956 & 281,403 & \textbf{11,878 }\\
V4 (4) & 2 & 3,771 & 2,036,833 & \textbf{ 40,334 } \\ 
\hline
\end{tabular}
\end{table}

The coefficients in the discovered clusters are relevant for premium calculations in auto insurance. Table~\ref{severity_coef_table} (Appendix~\ref{app:tables}) shows the coefficients of the fitted model. In each cluster, statistical significance varies but overall the majority of coefficients are statistically significant.

In summary, the drivers have been clustered into four categories with distinct characteristics outlined in Table~\ref{table:volSev}. We have seen how using the results from GCWM, one can create a rate model based on clustering algorithms with various levels of risk represented in each cluster. GCWM found a group that contains a clear majority of drivers, in which the volatility of their claims was extremely low regardless of $Density$ or $DriverAge$. The results show that GCWM may potentially discover unique clusters that are otherwise hidden within the data.

 \subsubsection{Modeling Claims Frequency}

In this section, we model frequency of the French motor claims. We consider the covariates $Density$, $DriverAge$, $CarAge$, and $Exposure$. Here, $Exposure$ is used as an offset to account for the rate at which a claim occurs \citep[see][]{frees2015}. The choice of covariates stems from the previously modelled single component ZIP \citep{Charpentier:2014}.
The ZI-GCWM is fitted with the following expression:
\begin{align}
g_P(\mathbb{E}\left[Y_{ClaimNb}|\bm{x}, \bm{\beta}_k \right]/x_{Exposure}) & = 
  \beta_{k0} +  \beta_{kDensity}x_{Density}+ \beta_{kDriverAge}x_{DriverAge} \nonumber \\    &  \quad\quad\quad +  \beta_{kCarAge}x_{CarAge}  =: \bm{\xTilda} \bm{\beta}_k^{'},   \label{poissonReg}\\
g_{ZI}(\mathbb{E}\left[Y_{ClaimNb}|\bm{x}, \bar{\bm{\beta}}_k \right]/x_{Exposure})& = \bar{\beta}_{k0} +  \bar{\beta}_{kDensity} x_{Density} =: \bm{\xTilda} \bar{\bm{\beta}}_k^{'},  \label{zeroReg} 
\end{align}
where $Density$, $DriverAge$, and $CarAge$ are explanatory variables in a Poisson model, while only $Density$ is an explantory variable for a Bernoulli model. As in Section 4.2.1, we also impose a log-normal assumption on the $Density$ covariate. The link functions $g_p$ and $g_z$ are chosen to be the log link and logit link respectively as in (\ref{g1link}) and (\ref{g2link}).
After fitting ZI-GCWM we find two zero-inflated components and one Poisson component as the best model to represent the data. The size of each cluster is displayed in Table~\ref{table:sizeFreq} and we note a fairly even spread of claims across Clusters~1 and~2, with Cluster 3 only consisting of $0.53 \%$ of the claims.
\begin{table}[!htb]
\centering
\caption{Size of clusters and their colours for the ZI-GCWM a model.}
\label{table:sizeFreq}
\begin{tabular}{rrrr}
\hline
Cluster 1   & Cluster 2  & Cluster 3   \\
\hline
$191,601$& $219,393$ & $2,175$ \\
green & red & blue  \\
\hline
\end{tabular}
\end{table}

Similarly to modelling severity, the ZI-GCWM finds clusters with unique characteristics. This is evident when looking at the Claims versus Density plot (Figure~\ref{frequencyGraph}). We see that the ZI-GCWM has assigned the drivers into three distinct groups based on the Density of cities. Table \ref{summarycovariates} shows that Cluster 2 drivers live in the most dense areas with a mean of 7.37 km on the log-scale. Followed by Clusters 3 and 1 with a mean of 5.45 km and 4.05 km, respectively.
\begin{figure}[!ht]
\begin{center}
\includegraphics[scale=0.6]{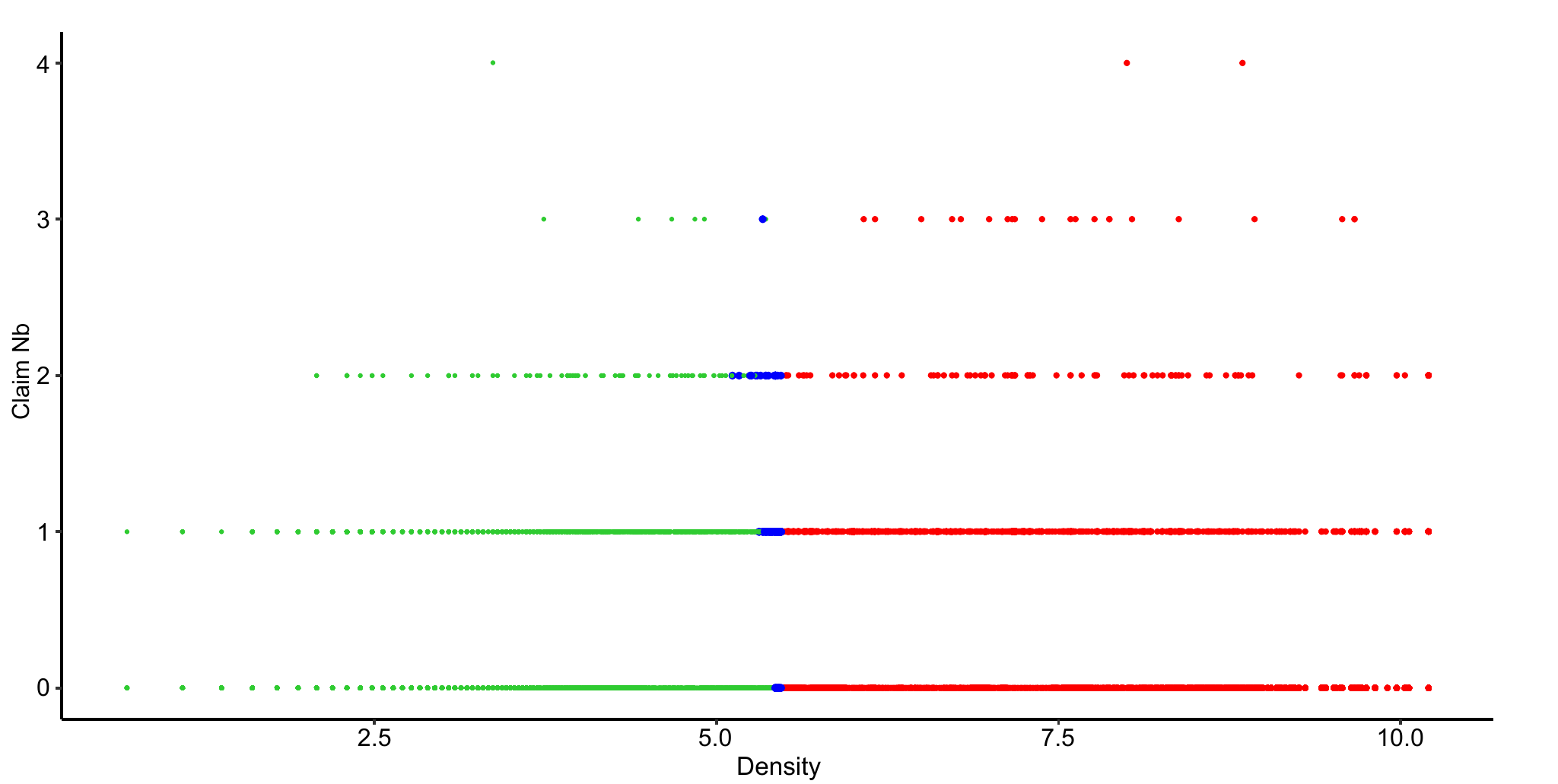}
\end{center}
\vspace{-0.2in}\caption{Showing clusters by color scheme --- Cluster~1 (green), Cluster~2 (red), Cluster~3 (blue) --- for $Claim Nb$ versus $Density$ on a log-scale.}
\label{frequencyGraph}
\end{figure}
\begin{table}[!htb]
 \begin{center}
 \caption{Summary of each cluster with log-normal assumptions for the $Density$ covariate measured in km on a log scale.} \label{summarycovariates}
\begin{tabular}{rlrrrr}
\hline
Cluster  & Color & Minimum & Mean & Maximum & $\sigma$  \\
\hline
1       & green  & 0.69 & 4.05 & 5.46  & 0.87 \\
2       & red    & 5.48 & 7.37 & 10.20 & 1.24 \\
3       & blue   & 5.12 & 5.45 & 5.47  & 0.03 \\
\hline
\end{tabular}
\end{center}
\end{table}

Table \ref{frequencySummary} (see Appendix~\ref{app:tables}) shows a summary of the coefficients for the zero-inflated model. The
significance codes are the same as in Table \ref{severity_coef_table}. In each cluster we can see that the majority of the coefficients are significant. Specifically for the coefficients pertaining to the the Bernoulli zero-count models. 

\begin{table}[!htb]
 \begin{center}
\caption{Comparison of BIC and Chi-square ($\chi^2$) test \eqref{LRTest} for CWM vs ZI-GCWM model on each cluster. }
\label{compareResults_models}
\begin{tabular}{r|rrc}
\hline
Cluster & BIC (CWM)  & BIC (ZI-GCWM) &  $\varphi $\\
\hline
 1 & $58,311$ & $\bm{58,180}$ & $ 4.61 < 154.70$ \\
 2 & $76,672$ & $\bm{76,301}$ & $ 4.61 < 395.59 $\\
 3 & $\bm{1,488}$  & $1,503$ &  -  \\ 
\hline 
All & $136,471$ & $135,984$ &  $10.64 < 270.03$
  \\
\hline
\end{tabular}
\end{center}
\end{table}
The results for comparing ZI-GCWM to CWM are shown in Table \ref{compareResults_models}. By using the likelihood ratio test in \eqref{LRTest} and comparing the models across each component, The Clusters~1 and 2 have evidence of zero inflation. In addition, ZI-GCWMs have BIC values that are lower in comparison to CWMs.  Cluster~3 has no evidence of zero-inflation which is in agreement with the BIC comparison. Thus, only a Poisson model is chosen for Cluster~3.  Overall, we see that there is evidence of zero inflation across the entire dataset when comparing ZI-GCWM versus CWM across all clusters. Thus, for rate making purposes, the ZI-GCWM can account for zero-inflation.

\section{Simulation Study}\label{sec:sim}

Two simulation studies are conducted to determine the validity of the log-normal assumption and the effectiveness of the Bernoulli-Poisson partitioning method. The first section outlines the need for a non-Gaussian assumption for the covariates. The second section shows the classification accuracy and other relevant analysis for the Bernoulli-Poisson method.

\subsection{GCWM Simulation Study}

In this section, we show how the proposed methodology performs under different simulation settings. The simulation study is generated based on the regression coefficients of the CASdataset used in the previous section. The aim of the simulation study is to test the accuracy and ability of both GCWM a and CWM to return estimates of true parameters when one or more of the covariates is log-normal and the other two are Gaussian. This design specifically tests both models in the event when one of the covariates is non-Gaussian. The motivation behind this choice lies in fact that many covariates used in insurance are likely to come from non-Gaussian distributions. Thus, this simulation is aimed to test the relevancy of CWM, which treats all covariates as Gaussian.

We define Model 1 as the baseline model, where the coefficients are selected to be reminiscent of those estimated from \textbf{CASdataset} and reported in upper portion of Table \ref{severity_coef_table} (see Appendix B). For the purposes of simplicity only three covariates are chosen and denoted as $X_1$ $X_2$ and $X_3$. The intercept for each component is increased to make sure simulated loss is positive. For ease of interpretation, these coefficients are then rounded and treated as true parameters for the simulation study.  A simulation with a 3 group mixture model is generated from the aforementioned true parameters however, the third covariate ($X_3$) is generated from a log-normal distribution.  In addition, the covariate $X_2$ for the second component is made insignificant and has no effect on the response.

Results aggregated from $1000$ runs are summarized in Table \ref{gcwmAccuracy} and Table \ref{mseTable}. Given a Gaussian assumption of for the residual error, we record the percentage of runs in which the error fall between a two-tailed $\% 5$ confidence interval in Table \ref{gcwmAccuracy}. For example, we report $90.10\%$ accuracy for predictor $X_2$ in the component 2. This means that $90.1\%$ of the time the true parameter is estimated within a $95\%$ confidence interval. 
Further, we create Models 2, 3, 4 and 5 by altering the parameters of Model~1 by $+30\%$, $-30\%$, $+50\%$, and $-50\%$, respectively, and keeping the second covariate of the second component as an insignificant predictor from the CASdataset model. This is done to test the accuracy of the GCWM and its sensitivity to changes in coefficient sizes. Based on the results in Table \ref{gcwmAccuracy}, we can see that GCWM performs well for all simulation settings.
\begin{table}[!htb]
\centering
\caption{GCWM a vs CWM Accuracy: covariate $X_3$ is treated as log-normally distributed, while the rest of covariates are of the Gaussian type.}
\label{gcwmAccuracy}
\begin{tabular}{ll|rrrr|rrrrr}
\hline
Mod. & $$ & Int. & $X_1$ &$X_2$ & $X_3$& Int. & $X_1$ &$X_2$ & $X_3$  \\
\hline
1     & 1         & 93.00\%   & 90.10\%  & 93.00\%  & 93.10\% & 0.00\% & 0.00\% & 0.00\% & 0.00\%   \\
      & 2         & 90.10\%   & \underline{0.00\%}   & 90.10\%  & 90.10\% & 0.00\% & 0.00\% & 0.00\% & 0.00\%  \\
      & 3         & 99.20\%   & 99.10\%  & 99.20\%  & 99.20\% & 0.00\% & 0.00\% & 0.00\% & 0.00\%  \\
      \hline
2     & 1         & 89.80\%   & 89.20\%  & 89.80\%  & 89.80\% & 0.00\% & 0.00\% & 4.60\% & 0.00\%  \\
      & 2         & 89.20\%   &\underline{0.00\%}   & 89.20\%  & 89.20\% & 0.00\% & \underline{0.00\%} & 0.00\% & 0.00\%   \\
      & 3         & 99.20\%   & 99.20\%  & 99.20\%  & 99.20\% & 0.00\% & 0.20\% & 1.70\% & 0.00\%  \\
      \hline
3     & 1         & 100.00\%  & 100.00\% & 100.00\% & 100.00\%  & 0.00\% & 0.00\% & 0.00\% & 0.00\% \\
      & 2         & 100.00\%  & \underline{0.00\%}   & 100.00\% & 100.00\% & 0.00\% & 0.00\% & 0.00\% & 0.00\% \\
      & 3         & 99.20\%   & 99.20\%  & 99.20\%  & 99.20\%  & 0.00\% & 0.00\% & 0.00\% & 0.00\%\\
      \hline
      4 & 1 & 88.60\% & 86.80\% & 88.60\% & 87.00\%  & 0.00\% & 0.00\% & 0.00\%  & 0.00\%  \\
  & 2 & 86.90\% &\underline{ 0.00\%}  & 86.90\% & 86.90\% & 0.00\% & \underline{0.00\%} & 0.00\%  & 0.00\%  \\
  & 3 & 99.20\% & 99.20\% & 99.20\% & 99.20\% & 0.00\% & 0.00\% & 0.00\%  & 0.00\% \\
      \hline
5 & 1 & 85.90\% & 84.90\% & 85.60\% & 85.90\% & 0.00\% & 0.00\% & 0.00\%  & 0.00\% \\
  & 2 & 85.00\% &\underline{ 0.00\%}  & 84.90\% & 84.90\% & 0.00\% & \underline{0.00\%} & 0.00\%  & 0.00\%  \\
  & 3 & 99.20\% & 99.20\% & 99.20\% & 99.20\% & 0.00\% & 0.20\% & 10.90\% & 0.00\% \\
      \hline
\end{tabular}
\end{table}

For CWM, in line with expectations, we note that barely any of the simulation runs are estimated correctly, as most of the results are zero. This means that the performance of CWM approach is poor in presence of one non-Gaussian covariate which in this case is a log-normal covariate. 

Table \ref{mseTable} provides the summary of mean squared errors (MSEs). The MSE is computed using the following formula, MSE $(\beta) = \frac{\sum_i^n (\beta_i - \hat\beta_i ) ^2}{n}$. Here, $n$ accounts for the number of simulation runs, $\beta$ is the true parameter of interest while $\hat{\beta}$ accounts for its estimate. In case of GCWM, the MSE for each parameter pertaining to each model are aggregated in Table~\ref{mseTable}. 
Overall, comparing Tables \ref{mseTable} and \ref{my-label} it is clear that GCWM outperforms CWM in both accuracy and MSE. 
\begin{table}[h!]
\centering
\caption{ GCWM results: the summary of MSE for all parameters used in five models. The covariate $X_3$ is treated as log-normal distributed, while the rest of covariates are Gaussian. These results correspond to same simulated runs as those in Table~\ref{gcwmAccuracy}.}
\label{mseTable}
\begin{tabular}{ll|rrrrrrrr}
\hline
Mod. & $K$ & $\beta_o$ &  MSE($\beta_o$)   &  $\beta_1$ & MSE($\beta_1$)& $\beta_2$ &MSE($\beta_2$)   & $\beta_3$ &  MSE($\beta_3$)  \\
\hline
1     & 1         & 1028& (11.353)   & 0.03& (0.00)  & 3.5& (0.00)    & -380& (0.09)   \\
      & 2         & 1600& (0.000)     & -0.01&(0.00) & 1.5&(0.00)    & -250&(0.00)   \\
      & 3         & 40000&(0.035)    & -6.00&(0.00) & -305&(0.00) & 1100&(0.47)   \\
\hline
2     & 1         & 1350&(0.167)     & 0.04&(0.00)  & 4.5&(0.00)    & -500&(0.03)   \\
      & 2         & 2080& (0.001)     & 0.04&(0.00)  & 2.0&(0.00)    & -325&(0.00)   \\
      & 3         & 52000& (0.012)    & -8.00&(0.00) & 450&(0.00)  & 14300&(0.01)  \\
\hline
3     & 1         & 720& (0.001)      & 0.02&(0.00)  & 2.5&(0.00)   & -266&(0.00)   \\
      & 2         & 1100& (0.008)     & 0.00&(0.00)  & 1.1&(0.00)    & -17511&(0.00) \\
      & 3         & 28000& (0.002)    & -4.20&(0.00) & 245&(0.00)  & 7700.&(0.00) \\
\hline
4     & 1         & 1650&(13.056)   & 0.05&(0.00)  & 5.3&(0.00)    & -570&(0.00)   \\
      & 2         & 2400& (0.000)     & -0.01&(0.00) & 2.3&(0.00)    & -375&(0.00)   \\
      & 3         & 60000& (0.051)    & -9.00&(0.00) & -457&(0.00) & 16500&(0.00)  \\
\hline
5     & 1         & 500& (1.115)     & 0.02&(0.00)  & 2.0&(0.00)    & -190&(0.05)   \\
      & 2         & 800& (0.003)      & 0.00&(0.00)  & 0.8&(0.00)    & -120&(0.00)   \\
      & 3         & 20000& (0.000)    & -3.00&(0.00) & -150&(0.00) & 5500&(0.00)  \\
      \hline
\end{tabular}
\end{table}
\begin{table}[!htb]
\centering
\caption{CWM results: the summary of MSE for all parameters used in five models. All three covariates are treated as Gaussian. These results correspond to same simulated runs as those in Table~\ref{gcwmAccuracy}.}
\label{my-label}
\begin{tabular}{ll|rrrrrrrr}
\hline
Mod. & $K$ & $\beta_o$ &  MSE($\beta_o$)   &  $\beta_1$ & MSE($\beta_1$)& $\beta_2$ &MSE($\beta_2$)   & $\beta_3$ &  MSE($\beta_3$)  \\
\hline
1     & 1         & 1028& ($\cdot$)   & 0.03&  ($\cdot$)   & 3.5&  ($\cdot$)    & -380&  ($\cdot$)    \\
      & 2         & 1600&  ($\cdot$)      & -0.01& ($\cdot$)  & 1.5& ($\cdot$)     & -250& ($\cdot$)  \\
      & 3         & 40000& ($\cdot$)     & -6.00& ($\cdot$)  & -305& ($\cdot$)  & 1100& ($\cdot$)    \\
\hline
2     & 1         & 1350& ($\cdot$)     & 0.04& ($\cdot$) & 4.5& ($\cdot$)    & -500& ($\cdot$)  \\
      & 2         & 2080&  ($\cdot$)    & 0.04& ($\cdot$)   & 2.0& ($\cdot$)     & -325& ($\cdot$)   \\
      & 3         & 52000&  ($\cdot$)     & -8.00& (0.006)  & 450& (44.1)   & 14300& ($\cdot$)  \\
\hline
3     & 1         & 720&  ($\cdot$)     & 0.02& ($\cdot$)   & 2.5& ($\cdot$)    & -266& ($\cdot$)    \\
      & 2         & 1100&  (65.814)     & 0.00& ($\cdot$)   & 1.1& ($\cdot$)     & -17511& ($\cdot$)  \\
      & 3         & 28000& ($\cdot$)   & -4.20& ($\cdot$)  & 245& ($\cdot$)   & 7700.& ($\cdot$)  \\
\hline
4     & 1         & 1650& ($\cdot$)    & 0.05& ($\cdot$)  & 5.3& ($\cdot$)    & -570& ($\cdot$)  \\
      & 2         & 2400&  ($\cdot$)     & -0.01& ($\cdot$)  & 2.3& ($\cdot$)    & -375& ($\cdot$)    \\
      & 3         & 60000&  ($\cdot$)     & -9.00& ($\cdot$)  & -457& ($\cdot$)  & 16500& ($\cdot$)   \\
\hline
5     & 1         & 500&  ($\cdot$)     & 0.02& ($\cdot$)   & 2.0& ($\cdot$)   & -190& ($\cdot$)  \\
      & 2         & 800&  ($\cdot$)      & 0.00& ($\cdot$)   & 0.8& ($\cdot$)    & -120& ($\cdot$)  \\
      & 3         & 20000&  ($\cdot$)     & -3.00& (0.003)  & -150& (4.7) & 5500& ($\cdot$) \\
\hline
\end{tabular}
\end{table}

From Table~\ref{my-label} we can observe that the MSEs for most of the models and their corresponding coefficients are not calculated at all due to convergence failures and as such they are shown as $(\cdot)$. This is not surprising because Table \ref{gcwmAccuracy} shows the accuracy of CWM is low when attempting to model non-Gaussian predictors as Gaussian. 

In summary, our simulation results showed good performance of the GCWM approach in modeling non-Gaussian covariates. More specifically, these results show high accuracy when covariates are log-normal. In contrary, CWM fails to estimate parameters accurately when the Gaussian assumption is violated.

\subsection{Bernoulli-Poisson Partitioning Simulation Study}

In this section we show how the Bernoulli-Poisson (BP) partitioning method behaves under different conditions. The components are genereated under similar coefficients estimated from the \textbf{CASDatasets} package. Again for easy of interpretation, coefficients are rounded and treated as true parameters from which the simulated data is generated from. The mean and standard deviation of the covariates within each component was also taken into account when generating data. The first simulation examines the performance of the ZI-GCWM model for classification. We generate three components each with sample size $N=1000$ for a total of $3000$ simulated points.
The model generated is similar in the mean and standard deviations presented in Table \ref{summarycovariates}. Consider three simulated covariates with 
\begin{align}
g_P(\mathbb{E}\left[Y_{SimClaimNb}|\bm{x}, \bm{\beta}_k \right]) & = 
  \beta_{k0} +  \beta_{kSimDensity} x_{SimDensity} +  \beta_{kSimDriverAge} x_{SimDriverAge} \nonumber  \\ & + 
   \beta_{kSimCarAge} x_{SimCarAge} := \bm{\xTilda} \bm{\beta}_k^{'},  \label{poissonRegSim} \\
g_{ZI}(\mathbb{E}\left[Y_{SimClaimNb}|\bm{x} , \bar{\bm{\beta}}_k  \right]) & = 
  \bar{\beta_{k0}} +  \bar{\beta}_{kSimDensity} x_{SimDensity} +  \bar{\beta}_{kSimDriverAge} x_{SimDriverAge} \nonumber \\ & + 
  \bar{ \beta}_{kSimCarAge} x_{SimCarAge} =: \bm{\xTilda} \bar{\bm{\beta}}_k^{'}.  \label{zeroRegSim}
\end{align}
as their respective linear models.  The covariates $x_{SimDensity}$, $x_{SimDriverAge}$, and $x_{SimCarAge}$ are considered for both the Poisson and Bernoulli models. Furthermore the link functions $g_P$ and $g_{ZI}$ are chosen to be log-link and logit-link respectively. 
 Here, the ZI-GCWM classifies drivers based on simulated data into three components. The misclassification rate is calculated by the proportion of true labels placed in other components by the ZI-GCWM model.  The results of the simulation are aggregated in Table \ref{misclassTable}. We observe the total misclassification rate of $1.8 \% $, where the makority of misclassified components are between components two and three.
\begin{table}[!htb]
\begin{center}
\caption{Misclassfication rate and label comparison of generated data.}
\label{misclassTable}
\begin{tabular}{r| r r r| r r}
\hline
    True Labels       &  \multicolumn{3}{r}{ Classified }  \vline & Misclassification Rate  &  \\ \hline
   & 1                              & 2   & 3   &                            &  \\ \hline
 1              & 992                            & 3   & 5   & 0.80 \%                                      &  \\
 2              & 0                              & 990 & 10  & 1.00 \%                                       &  \\
 3              & 15                             & 20  & 965 & 3.50 \%                                      &  \\  \hline
                \multicolumn{4}{r}{Overall Misclassification Rate}        & 1.80 \%                  & \\
        		\multicolumn{4}{r}{Average Purity} & 98.23 \%  
        		\\ \hline
                \multicolumn{4}{r}{Adjusted Rand Index} & 0.9479 &  \\
    \hline
\end{tabular}
\end{center}
\end{table}

The simulation is expanded further to show how BP partitioning behaves over 1000 runs and under two different conditions. The first condition is defined as follows. The mean and standard deviations of covariates are taken directly from sample statistics of \textbf{CASDataset}. The second condition involves adjusting the means of two of the covariates so they are closer to each other. The goal is to show that the BP-method holds its use even when means among covariates are close. The conditions are divided into two scenarios. In the first scenario which we consider ``normal", the covariate means are taken directly from the sample data. In the second scenario defined as ``close", the covariate means are manipulated so that they are $20 \%$ closer to each other. This is a common problem in classification where if the means among two different components are close, then misclassification rate increases \citep{LimHwa}. The expansion of the simulation to $1000$ runs tests the accuracy of 3 different partitioning methods to initialize a zero-inflated model. The results of this expansion are aggregated in Table \ref{table:exper2}.  The Poisson partitioning method assumes that the presence of non-zeros will provide a better partitioning of the data-set. The Bernoulli partitioning method assumes that the presence of excess zeros will determine the best partitioning of the data-set. Finally, the BP partitioning method assumes that both methods are weighed equally and therefore both must be taken into account when partitioning the dataset. The mean and standard deviation of each measurement is provided in Table~\ref{table:exper2}.
\begin{table}[!htb]
\begin{center}
\caption{Aggregated results for the $1000$ run simulation, mean and standard deviations for each statistic are compared across three methods.}
\label{table:exper2}
{\small\begin{tabular}{llrrrrrr}
\hline
Type   & Condition & Poisson & ($\sigma $) & Bernoulli & ($ \sigma $) & BP & ($ \sigma $) \\
\hline
Misclassification Rate& normal        & 1.70\% & (6.00)       & 1.60\%  & (6.00)         & 1.10\% & (0.02)         \\
       & close      & 5.00\% & (7.00)       & 6.00\% & (2.00)         & 7.00\% & (4.00)         \\
Average Purity & normal     & 98.87\% & (2.00)    & 98.91\% & (2.25)      & 99.18\% & (0.81)     \\
       & close       & 95.38\% & (4.00)    & 94.55\% & (1.00)      & 96.95\% & (0.48)      \\
Adjusted Rand Index  & normal      & 0.9662 & (0.07)    & 0.9677  & (0.07)     & 0.9729 & (0.0217)      \\
       & close         & 0.8706 & (0.08)    & 0.8366 & (0.04)      & 0.8538 & (0.0453) \\
       \hline
\end{tabular}}
\end{center}
\end{table}

 Under condition normal, the BP method shows better performance in error and is found to be less sensitive than other methods with an error rate of $ 1.10 \% $ and a standard deviation of $ 0.02 \% $.  Furthermore, when the close condition is imposed, the partitioning using only the Bernoulli method has better performance in terms of accuracy. The adjusted Rand index \citep[ARI][]{hubert85} is the Rand index \citep{rand71} corrected for chance agreement. The Rand index, for two partitions, is simply the number of pair agreements divided by the total number of pairs. The ARI takes a value 1 for perfect class agreement and has expected value 0 under random classification.
 The ARI measurements across all methods are promising. In particular the BP partitioning method under the normal condition has a very good ARI with a small standard deviation. The Average Purity (AP) is calculated as the average of the diagonal classification entries. 
The AP for the BP partitioning method is the best out of all other methods, therefore the BP method is the most relevant for classification.

\section{Conclusion}

By accomplishing two main goals, in this work, we extend the class of generalized linear mixture CWM models to ZI-GCWM models. First, we proposed the methodology that allows for continuous covariates to follow a non-Gaussian distribution. This is an important extension, at least in insurance modeling context, as imposing Gaussian distribution on a skewed data may result in an suboptimal model fit. Second, we proposed a new CWM methodology that uses BP partitioning method and allows for implementation of zero-inflated CWM. 


Our proposed GCWM models allow applications in predictive modeling of insurance claims by overcoming a few limitations of the current CWM models. The GCWM allows for finding clusters within claims frequency which is an important information in risk classification and modeling of claims frequency. Further, some insurance rating variables used in the predictive modeling of severity claims may not strictly follow Gaussian assumptions, for example driver's age or car age, when treated as continuous covariates. An adequate extension to non-Gaussian covariates can be considered to relax current assumptions and improve the model fit. Given our data, we convincingly demonstrated that there is a need for a log-normal assumption in the $Density$ covariate, and by making it we have considerably improved the model fit. 

The results of our extensive simulation study showed the excellent performance of the proposed models in case of modeling non-Gaussian covariates. We found that the CWM model fails to estimate the parameters accurately when the Gaussian assumption is violated. The GCWM a shows significant improvement in the model fit over the CWM model based on AIC and BIC criteria. We also tested BP partitioning of zero-inflated GCWM under different conditions and found that our proposed partitioning method has a very low misclassification rate, high average purity, and high average ARI. Our approach is highly relevant to actuarial pricing and risk management, where current practices are based on implementation of various GLM models.

\appendix
\section{Multivariate Log-Normal Distribution }
Consider a random variable $U$ having univariate log-normal distribution with parameters $\mu \in \mathbb{R}$ and $\sigma \in \mathbb{R}_+ $. Have $u \in \mathbb{R}_+$, then the probability density function of random variable $U$ is defined as \footnote{For full definition see \cite{johnson1995continuous}}
$$\mathcal{LN}(u; \mu, \sigma) = \frac{1}{u\sigma\sqrt{2\pi}}\exp\left[-\frac{(\ln u - \mu)^2}{2\sigma^2}	\right].$$
\text{Further, if random variable }$X$\text{ is normally distributed i.e. }$ X \sim \mathcal{N}(x; \mu, \sigma) $, then $U := \exp{(X)}\sim \mathcal{LN}(u; \mu, \sigma) $.
To see this, let $p_U(u)$, and $ p_X(x) $ be the probability density functions of $U$ and $X$ respectively. By the change of variables theorem (see \cite{murphy2012machine} section 2.6.2.1) the density $p_U(u)$ is derived as
$$p_U(u) = p_X(\ln u )\frac{\partial}{\partial u} \ln u  =  p_X(\ln u ) \frac{1}{u} =  \frac{1}{u\sigma\sqrt{2\pi}}\exp\left[-\frac{(\ln u - \mu)^2}{2\sigma^2}	\right].$$\newline
 We extend to a log-normal multivariate case where the random variable $\bm{U} $ is parameterized by $ \bm{\mu} \in \mathbb{R}^p$ and $\bm{\Sigma} \in  \mathbb{R}_{+}^{p \times p} \label{changeVarUni} $.
\begin{lemma}
Let the random variable $\bm{X}$ have multivariate normal distribution ie. $\bm{X} \sim \mathcal{MVN}(\bm{x}, \bm{\mu},\bm{\Sigma}) $, then $\bm{U} := \exp(\bm{X} ) \sim  f^U(\bm{u}; \bm{\mu } , \bm{\Sigma} )$. Here
have $\bm{u} \in \mathbb{R}_{+}^p $ and the probability density function $f^U$ is
$$ f^U(\bm{u}; \bm{\mu } , \bm{\Sigma} )= \frac{1}{(\prod_{i=1}^{p}u_{i})| \bm{\Sigma} |(2 \pi)^{\frac{p}{2}}}   \exp\left[-\frac{1}{2}(\ln \bm{u} -\bm{\mu})^{'}  \bm{\Sigma}^{-1}(\ln \bm{u} -\bm{\mu})\right].  $$
\end{lemma}
\begin{proof}
Let $f^U(\bm{u}; \bm{\mu},\bm{\Sigma})$ and $f^X(\bm{x}; \bm{\mu},\bm{\Sigma})$ be the probability density functions of $\bm{U}$ and $\bm{X}$ respectively. By the multivariate change of variables theorem \citep[see][Section~2.6.2.1]{murphy2012machine}, we derive the log-normal distribution, where $ | \det K_{\ln} (u) | $ is the absolute value of the determinant for the Kacobian of the multivariate transformation $\ln(\bm{U}) = \bm{X} $. Hence,
\begin{align*}
 | \det K_{\ln} (\bm{u}) | & = \prod_{i=1}^p u_i^{-1}, \; \text{and} \; \\
   f^U(\bm{u}; \bm{\mu},\bm{\Sigma})  & =  f^X(\ln \bm{u}; \bm{\mu},\bm{\Sigma})  | \det K_{\ln} (u) | \\
  & = f^X(\ln \bm{u}; \bm{\mu},\bm{\Sigma})\prod_{i=1}^p u_i^{-1} \\
  & =  \frac{1}{(\prod_{i=1}^{p}u_{i})| \bm{\Sigma} |(2 \pi)^{\frac{p}{2}}}   \exp\left[-\frac{1}{2}(\ln \bm{u} -\bm{\mu})^{'}  \bm{\Sigma}^{-1}(\ln \bm{u} -\bm{\mu})\right].
  \end{align*}
\end{proof}

\begin{center}
\begin{sidewaystable}

\section{Tables}\label{app:tables}

\caption{ Summary of coefficients for severity clusters.}
\label{severity_coef_table}
{\small\begin{tabular}{l|rrc|rrc|rrc|rrc}
\hline
         & V1         & (blue)     &    & V2         & (green)   &    & V3          & (red)     &    & V4          & (orange)     &    \\
Coefficient \footnote{The significance codes are defined as $  P < 0.001 : $  (***), $0.001 < P < 0.01:$ (**), $  0.01 < P < 0.05:$ (*),\\ $0.05 < P < 0.10 : $ (.) 
pertaining to the $P$ value of the specific coefficient. C\# refers to the Car Age category, D\# refers to the Driver Age category, R\# refers to the region of France, and P refers to the power category.  }      & Estimate   & Error     & P   & Estimate   & Error     & P   & Estimate    & Error      & P   & Estimate    & Error      & P   \\ \hline
Int. & 7.077 & 0.003 & *** & 6.952 & 0.043 & *** & 7.306 & 0.138 & *** & 7.212 & 0.136 & *** \\
Density & 0.000 & 0.000 &  & -0.009 & 0.003 & ** & -0.006 & 0.011 &  & -0.052 & 0.014 & *** \\
C2 & 0.008 & 0.002 & *** & 0.069 & 0.023 & ** & -0.278 & 0.064 & *** & 0.329 & 0.074 & *** \\
C3 & 0.002 & 0.002 &  & 0.222 & 0.023 & *** & -0.460 & 0.064 & *** & 0.161 & 0.074 & * \\
C4 & 0.004 & 0.002 & . & 0.075 & 0.024 & ** & -0.693 & 0.066 & *** & 0.103 & 0.074 &  \\
C5 & 0.009 & 0.002 & *** & 0.102 & 0.027 & *** & -0.608 & 0.076 & *** & 0.234 & 0.081 & ** \\
D2 & -0.007 & 0.002 & *** & 0.031 & 0.026 &  & -0.210 & 0.080 & ** & -0.690 & 0.068 & *** \\
D3 & -0.008 & 0.002 & *** & -0.021 & 0.026 &  & -0.250 & 0.081 & ** & -0.834 & 0.069 & *** \\
D4 & -0.012 & 0.002 & *** & -0.014 & 0.031 &  & -0.122 & 0.091 &  & -0.753 & 0.084 & *** \\
D5 & -0.006 & 0.002 & ** & 0.078 & 0.032 & * & 0.108 & 0.096 &  & -0.182 & 0.083 & * \\
R23 & 0.002 & 0.004 &  & -0.059 & 0.036 & . & 0.115 & 0.110 &  & -0.007 & 0.122 &  \\
R24 & -0.013 & 0.001 & *** & 0.091 & 0.016 & *** & -0.279 & 0.042 & *** & -0.003 & 0.075 &  \\
R25 & -0.019 & 0.002 & *** & -0.362 & 0.030 & *** & -0.027 & 0.086 &  & 0.257 & 0.099 & ** \\
R31 & -0.002 & 0.002 &  & 0.025 & 0.020 &  & 0.111 & 0.053 & * & 0.035 & 0.106 &  \\
R52 & -0.016 & 0.002 & *** & -0.002 & 0.019 &  & -0.260 & 0.051 & *** & 0.015 & 0.085 &  \\
R53 & -0.013 & 0.002 & *** & 0.119 & 0.019 & *** & -0.106 & 0.053 & * & 0.092 & 0.082 &  \\
R54 & -0.014 & 0.002 & *** & 0.099 & 0.026 & *** & -0.295 & 0.072 & *** & 0.117 & 0.090 &  \\
R72 & -0.008 & 0.002 & *** & 0.123 & 0.021 & *** & 0.003 & 0.056 &  & 0.239 & 0.088 & ** \\
R74 & -0.020 & 0.003 & *** & -0.125 & 0.050 & * & -0.141 & 0.170 &  & 0.131 & 0.118 &  \\
P-FGH & 0.001 & 0.001 & . & 0.006 & 0.011 &  & 0.108 & 0.030 & *** & 0.003 & 0.030 &  \\
P-Other & 0.005 & 0.001 & *** & 0.013 & 0.014 &  & 0.116 & 0.038 & ** & 0.057 & 0.041 &  \\
\hline
\end{tabular}}
\end{sidewaystable}

\begin{sidewaystable}
\caption{Summary of coefficients for frequency clusters.}
\label{frequencySummary}
\begin{tabular}{l|rrc|rrc|rrc}
\hline
          & Cluster 1 & (green) &   & Cluster 2 & (red) &  & Cluster 3 & (blue) &   \\
Coefficient \footnote{The significance codes are defined as $  P < 0.001 : $  (***), $0.001 < P < 0.01:$ (**), $  0.01 < P < 0.05:$ (*),\\ $0.05 < P < 0.10 : $ (.) 
pertaining to the $P$ value of the specific coefficient. C\# refers to the Car Age category, D\# refers to the Driver Age category.  }         & Estimate  & Error & P   & Estimate  & Error   & P   & Estimate  & Error  & P  \\
 \hline
Intercept & -4.71437 & 0.16072 & *** & -4.21884 & 0.23567 & *** & 80.34056  & 4.82835 & *** \\
Density   & 0.37081  & 0.03414 & *** & 0.20128  & 0.02916 & *** & -15.36444 & 0.87135 & *** \\
D2        & -0.09767 & 0.05987 &     & -0.15287 & 0.06008 & *   & -0.28067  & 0.3133  &     \\
D3        & -0.28937 & 0.06156 & *** & -0.30113 & 0.06088 & *** & -0.33352  & 0.31412 &     \\
D4        & -0.02955 & 0.08005 &     & 0.03986  & 0.0711  &     & -1.02815  & 0.39399 & **  \\
D5        & 0.6164   & 0.07787 &     & 0.47363  & 0.07558 & *** & 0.11692   & 0.46536 &     \\
C2        & 0.67169  & 0.67169 & *** & 0.59803  & 0.59803 & *** & -1.00529  & 0.79119 &     \\
C3        & 0.69215  & 0.69215 & *** & 0.85653  & 0.85653 & *** & 2.05854   & 0.71373 & **  \\
C4        & 0.63158  & 0.63158 & *** & 0.76843  & 0.76843 & *** & 2.20552   & 0.71347 & **  \\
C5        & 0.36033  & 0.36033 & *** & 0.52438  & 0.52438 & *** & 2.06543   & 0.72307 & **  \\
\hline
Intercept & -3.9712  & 0.5473  & *** & -1.66782 & 0.37735 & *** &           &         &     \\
Density   & 0.9032   & 0.1041  & *** & 0.28258  & 0.04674 & *** &           &         &    \\
\hline
\end{tabular}
\end{sidewaystable}
\end{center}
\newpage
\section*{Acknowledgements}
The authors express sincere gratitude to Dr. Paul Wilson at the School of Mathematics and Computer Science, University of Wolverhampton for his kind help and advice. In addition, the authors wholeheartedly thank the author and maintainer of the French Motor Policy dataset, Dr. Christophe Dutang at the  Universit\' e Paris Dauphine, for his generous support. Finally, the authors thank Dr. Ben Bolker at the Mathematics and Statistics Department, McMaster University for his kind advice.


\begin{thebibliography}{36}
\expandafter\ifx\csname natexlab\endcsname\relax\def\natexlab#1{#1}\fi
\expandafter\ifx\csname url\endcsname\relax
  \def\url#1{\texttt{#1}}\fi
\expandafter\ifx\csname urlprefix\endcsname\relax\def\urlprefix{URL }\fi

\bibitem[{Berm\'{u}dez and Karlis(2012)}]{Bermudez+Karlis:2012}
Berm\'{u}dez, L., Karlis, D., 2012. A finite mixture of bivariate poisson
  regression models with an application to insurance ratemaking. Computational
  Statistics and Data Analysis 56~(12), 3988--3999.

\bibitem[{Biernacki et~al.(2000)Biernacki, Celeux, and
  Govaert}]{initialPaperGrassiaRef}
Biernacki, C., Celeux, G., Govaert, G., 2000. Assessing a mixture model for
  clustering with the integrated completed likelihood. IEEE Transactions on
  Pattern Analysis and Machine Intelligence 22~(7), 719--725.

\bibitem[{Brown and Buckley(2015)}]{Brown+Buckley:2015}
Brown, G.~O., Buckley, W.~S., 2015. Experience rating with poisson mixtures.
  Annals of Actuarial Science 9~(02), 304--321.

\bibitem[{Charpentier(2014)}]{Charpentier:2014}
Charpentier, A., 2014. Computational Actuarial Science with R. CRC press.

\bibitem[{Dempster et~al.(1977)Dempster, Laird, and
  Rubin}]{Dempster+Laird+Rubin:1977}
Dempster, A.~P., Laird, N.~M., Rubin, D.~B., 1977. Maximum likelihood from
  incomplete data via the {EM}-algorithm. Journal of the Royal Statistical
  Society B 39, 1--38.

\bibitem[{Dutang and Charpentier(2016)}]{Dutang+Charpentier:2016}
Dutang, C., Charpentier, A., 2016. {CASdatasets}. {R} package version 1.0-6.

\bibitem[{Frees et~al.(2014)Frees, Derrig, and Meyers}]{frees2015}
Frees, E.~W., Derrig, R.~A., Meyers, Glenn, e., 2014. Predictive Modeling
  Applications in Actuarial Science. Vol.~1 of International Series on
  Actuarial Science. International Series on Actuarial Science.

\bibitem[{Gershenfeld(1997)}]{Gershenfeld:1997}
Gershenfeld, N., 1997. Nonlinear inference and cluster-weighted modeling.
  Annals of the New York Academy of Sciences 808~(1), 18--24.

\bibitem[{Gershenfeld et~al.(1999)Gershenfeld, Schoner, and
  Metois}]{Gershenfeld:Schoner+Metois:1999}
Gershenfeld, N., Schoner, B., Metois, E., 1999. Cluster-weighted modelling for
  time-series analysis. Nature 397~(67171), 329--332.

\bibitem[{Gershenfeld(1999)}]{Gershenfeld:1999}
Gershenfeld, N.~A., 1999. The nature of mathematical modeling. Cambridge
  university press.

\bibitem[{Hubert and Arabie(1985)}]{hubert85}
Hubert, L., Arabie, P., 1985. Comparing partitions. Journal of Classification
  2~(1), 193--218.

\bibitem[{Ingrassia et~al.(2014{\natexlab{a}})Ingrassia, Minotti, and
  Punzo}]{Ingrassia+Minotti+Punzo:2014}
Ingrassia, S., Minotti, S.~C., Punzo, A., 2014{\natexlab{a}}. Model-based
  clustering via linear cluster-weighted models. Computational Statistics \&
  Data Analysis 71, 159--182.

\bibitem[{Ingrassia et~al.(2014{\natexlab{b}})Ingrassia, Minotti, and
  Vittadini}]{Ingrassia+Minotti+Vittadini:2012}
Ingrassia, S., Minotti, S.~C., Vittadini, G., 2014{\natexlab{b}}. Local
  statistical modeling via a cluster-weighted approach with elliptical
  distributions. Journal of classification 29~(3), 363--401.

\bibitem[{Ingrassia et~al.(2015)Ingrassia, Punzo, Vittadini, and
  Minotti}]{Ingrassia+Punzo+Vittadini+Minotti:2015}
Ingrassia, S., Punzo, A., Vittadini, G., Minotti, S.~C., 2015. Erratum to: The
  generalized linear mixed cluster-weighted model. Journal of Classification
  32~(2), 327--355.

\bibitem[{Johnson et~al.(1994)Johnson, Kotz, and
  Balakrishnan}]{johnson1995continuous}
Johnson, N.~L., Kotz, S., Balakrishnan, N., 1994. Continuous Univariate
  Probability Distributions,(Vol. 1). John Wiley \& Sons Inc., NY.

\bibitem[{Lambert(1992)}]{Lambert}
Lambert, D., 1992. Zero-inflated poisson regression, with an application to
  defects in manufacturing. Technometrics 34, 1--14.

\bibitem[{Lee and Lin(2010)}]{Lee+Lin:2010}
Lee, S. C.~K., Lin, X.~S., 2010. Modeling and evaluating insurance losses via
  mixtures of {E}rlang distributions. North American Actuarial Journal 14~(1),
  107--130.

\bibitem[{Lim et~al.(2014)Lim, Li, and Yu}]{LimHwa}
Lim, H., Li, W., Yu, P., 2014. Zero-inflated poisson regression mixture model.
  Computational Statistics and Data Analysis 71, 151--158.

\bibitem[{McCullagh and Nelder(1989)}]{McCullaghNelder1989}
McCullagh, P., Nelder, J.~A., 1989. Generalized linear models. Vol.~37. CRC
  press.

\bibitem[{McNicholas(2016)}]{mcnicholas16a}
McNicholas, P.~D., 2016. Mixture Model-Based Classification. Chapman \&
  Hall/CRC Press, Boca Raton.

\bibitem[{McNicholas and Murphy(2008)}]{McNicholas:2010}
McNicholas, P.~D., Murphy, T.~B., 2008. Parsimonious gaussian mixture models.
  Statistics and Computing 18~(3), 285--296.

\bibitem[{Miljkovic and Fern{\'a}ndez(2018)}]{risks_miljkovic}
Miljkovic, T., Fern{\'a}ndez, D., 2018. On two mixture-based clustering
  approaches used in modeling an insurance portfolio. Risks 6~(2), 57.

\bibitem[{Miljkovic and Gr\"un(2016)}]{Miljkovic+Grun:2016}
Miljkovic, T., Gr\"un, B., 2016. Modeling loss data using mixtures of
  distributions. Insurance: Mathematics and Economics 70, 387--396.

\bibitem[{Murphy and Bach(2012)}]{murphy2012machine}
Murphy, K.~P., Bach, F., 2012. Machine Learning: A Probabilistic Perspective.
  Adaptive Computation and Machi. MIT Press.

\bibitem[{Punzo and Ingrassia(2014)}]{Punzo+Ingrassia:2015}
Punzo, A., Ingrassia, S., ., 2014. Parsimonious generalized linear Gaussian
  cluster-weighted models. Springer International Publishing.

\bibitem[{Punzo and McNicholas(2017)}]{punzo17}
Punzo, A., McNicholas, P.~D., 2017. Robust clustering in regression analysis
  via the contaminated gaussian cluster-weighted model. Journal of
  Classification 34~(2), 249--293.

\bibitem[{{R Core Team}(2018)}]{R18}
{R Core Team}, 2018. R: A Language and Environment for Statistical Computing. R
  Foundation for Statistical Computing, Vienna, Austria.

\bibitem[{Rand(1971)}]{rand71}
Rand, W.~M., 1971. Objective criteria for the evaluation of clustering methods.
  Journal of the American Statistical Association 66~(336), 846--850.

\bibitem[{Subedi et~al.(2013)Subedi, Punzo, Ingrassia, and
  McNicholas}]{subedi13}
Subedi, S., Punzo, A., Ingrassia, S., McNicholas, P.~D., 2013. Clustering and
  classification via cluster-weighted factor analyzers. Advances in Data
  Analysis and Classification 7~(1), 5--40.

\bibitem[{Subedi et~al.(2015)Subedi, Punzo, Ingrassia, and
  McNicholas}]{subedi15}
Subedi, S., Punzo, A., Ingrassia, S., McNicholas, P.~D., 2015. Cluster-weighted
  t-factor analyzers for robust model-based clustering and dimension reduction.
  Statistical Methods and Applications 24~(4), 623--649.

\bibitem[{Verbelen et~al.(2015)Verbelen, Gong, Antonio, Badescu, and
  Lin}]{Verbelen+Gong+Antonio+Badescu+Lin:2015}
Verbelen, R., Gong, L., Antonio, K., Badescu, A., Lin, S., 2015. Fitting
  mixtures of {E}rlangs to censored and truncated data using the {EM}
  algorithm. ASTIN Bulletin 45~(3), 729--758.

\bibitem[{Vuong(1989)}]{vuongTest}
Vuong, Q.~H., 1989. Likelihood ratio tests for model selection and non-nested
  hypotheses. Econometrica 57~(2), 307--333.

\bibitem[{Wedel(2002)}]{Wedel:2002}
Wedel, M., 2002. Concominat variables in finite mixture modeling. Statistica
  Neerlandica 56~(3), 362--375.

\bibitem[{Wedel and De~Sabro(1995)}]{Wedel+DeSabro:1995}
Wedel, M., De~Sabro, W., 1995. A mixture likelihood approach for generalized
  linear models. Journal of Classification 12~(3), 21--55.

\bibitem[{Wilson(2015)}]{misuse}
Wilson, P., 2015. The misuse of the vuong test for non-nested models to test
  for zero-inflation. Economics Letters 127.

\bibitem[{Wilson and Einbeck(2018)}]{newIntuitive}
Wilson, P., Einbeck, J., 2018. A new and intuitive test for zero modification.
  Statistical Modelling.

\end{thebibliography}
\end{document}